\documentclass[12pt]{article}
\usepackage[utf8]{inputenc}
\usepackage[english]{babel}
\usepackage{a4wide,cite,hyperref,authblk}
\usepackage{graphicx}
\usepackage{color}
\usepackage{amsmath}
\usepackage[version=3]{mhchem}
\usepackage{amssymb}
\usepackage{amsthm}
\usepackage{enumerate}
\usepackage{txfonts}
\usepackage[fleqn,tbtags]{mathtools}
\usepackage{array}
\usepackage{framed}
\usepackage{fancybox}
\usepackage{comment}
\usepackage{ifthen}
\usepackage{fancyhdr}
\usepackage{scalerel,stackengine}
\usepackage{todonotes}
\stackMath
\newcommand\reallywidehat[1]{%
\savestack{\tmpbox}{\stretchto{%
  \scaleto{%
    \scalerel*[\widthof{\ensuremath{#1}}]{\kern-.6pt\bigwedge\kern-.6pt}%
    {\rule[-\textheight/2]{1ex}{\textheight}}
  }{\textheight}%
}{0.5ex}}%
\stackon[1pt]{#1}{\tmpbox}%
}

 {\endMakeFramed}

\newcounter{thm}
\numberwithin{equation}{section}
\numberwithin{thm}{section}

\newtheorem{lemma}{Lemma}[section]
\newtheorem{prop}{Proposition}[section]
\newtheorem{corollary}{Corollary}[section]

\newtheorem{theorem}[thm]{Theorem}

\theoremstyle{definition}

\newtheorem*{remark*}{Remark}
\newtheorem*{example*}{Example}
\newtheorem{proofpart}{Part}[thm]

\newcommand{\linePage} {\noindent\makebox[\linewidth]{\rule{\textwidth}{1pt}} \\}
\specialcomment{simple} {$ $ \\ \begingroup   \scriptsize Simple Calculation \vspace{-4mm}   \\ \linePage   \endgroup \begingroup   }{$ $ \vspace{-3mm} \\ \linePage  \endgroup}

\excludecomment{simple}

\begin{document}

\setlength{\headheight}{15pt}
\pagestyle{fancy}
\fancyhf{}
\rhead{Page \thepage}
\lhead{Thomas Moser, Robert Seiringer}
\newcounter{ct}

\newcommand{\marginnote}[1]{\marginpar{
\raggedright\footnotesize
\itshape#1\par}}

\newcommand{\lecture}[1] {\marginnote{\framebox{$\begin{matrix} \text{Mitschrift} \\ \text{#1} \end{matrix}$}}}

\newcommand{\bemerkung}[1] {\marginnote{#1}}

\newcommand{\ppd}[2] {\frac{\pd^{#2}}{\pd {#1}^{#2}}}
\newcommand{\ppdd}[1] {\frac{\pd}{\pd #1}}
\newcommand{\ddb} [2] {\frac{\mathrm{d} #1}{\mathrm{d} #2}}
\newcommand{\ddbb} [1] {\ddb{} {#1}}
\newcommand{\dd} {\mathrm{d}}
\newcommand{\dI} { \,\mathrm{d}}

\newcommand{\Angstrom}{\overset{\circ}{\text{A}}}
\newcommand{\Celsius}{^\circ \text{C}}
\newcommand{\Hil} {\mathcal H}
\newcommand{\Ident} {\mathds 1}
\newcommand{\tr} {\operatorname{tr}}
\newcommand{\bra}[1] { \langle #1 | }
\newcommand{\ket}[1] { | #1 \rangle }
\newcommand{\braket}[2]{\langle #1 | #2 \rangle}
\newcommand{\vect}[1]{\ensuremath{\boldsymbol{#1}}}
\newcommand{\ev}[1]{ \lk \langle #1 \rk \rangle}
\newcommand{\Var}[1]{ \operatorname{Var}\lk( #1 \rk)}
\newcommand{\norm}[1]{\left\| #1 \right\|}

\newcommand{\esssup}{\operatorname{ess}\operatorname{sup}}
\newcommand{\diver}{\operatorname{div}}
\newcommand{\e} {\operatorname{e}}
\newcommand{\supp} {\operatorname{supp}}
\newcommand{\STAR} {\operatorname{STAR}}

\newcommand{\R}{\mathbb{R}}
\newcommand{\C}{\mathbb{C}}
\newcommand{\LR}{\mathcal{L}}
\newcommand{\N} {\mathbb N}
\newcommand{\pd} {\partial}
\renewcommand{\Re}{\operatorname{Re}}
\renewcommand{\Im}{\operatorname{Im}}

\newcommand{\CC} {\mathcal C}
\newcommand{\K} {\mathcal K}
\newcommand{\E} {\mathcal E}
\newcommand{\wt}[1]{\widetilde{#1}}
\newcommand{\ol}[1]{\overline{#1}}
\newcommand{\ob}[2]{\overbrace{#1}^{#2}}
\newcommand{\ub}[2]{\underbrace{#1}_{#2}}
\newcommand{\ve}{\varepsilon}
\newcommand{\vp}{\varphi}
\newcommand{\Lap}{\Delta}
\newcommand{\rk} {\right}
\newcommand{\lk} {\left}
\newcommand{\intm}{\fint} 
\newcommand{\hateq}{ {\hat =} }

\newcommand{\tlim}{\lim\limits}
\newcommand{\transp}[1]{#1^{\mathrm t}}
\newcommand{\sprungk}[1]{\bigl[[#1]\bigr]}

\newcommand{\Oo}{{\mathcal O}}
\newcommand{\const}{{\rm const.}}

\newcommand{\usetc}[1] {\underset{\mathclap{\substack{#1}}}}
\newcommand{\osetc}[1] {\overset{\mathclap{\substack{#1}}}}

\newcommand{\usetceq}[1] {\underset{\mathclap{\substack{\uparrow \\ \eqref{#1}}}}}
\newcommand{\usetceqb}[1] {\underset{\mathclap{\substack{\big \uparrow \\ \eqref{#1}}}}}
\newcommand{\usetceqB}[1] {\underset{\mathclap{\substack{\Big \uparrow \\ \eqref{#1}}}}}

\newcommand{\setbackEnum} {  $ $ \vspace*{-4mm}}
\newcommand{\inwork}{$ $ \\ \begin{center}\Ovalbox{\large\bf In Work} \end{center} $ $ \\}

\newcommand{\folgt} {\ \Rightarrow \ }
\newcommand {\numberthis} {\addtocounter{equation}{1}\tag{\theequation}}

\newcommand{\noeqspace}{\vspace{-10mm}}
\setcounter{tocdepth}{3}
\setcounter{secnumdepth}{3}

\newcommand{\calF}{\mathcal F}
\newcommand{\calT}{\mathcal T}
\newcommand{\calG}{\mathcal G}
\newcommand{\calL}{\mathcal L}

\newcommand{\rel}{{\rm rel}}
\newcommand{\per}{{\rm per}}

\newcommand{\Tdiag}{T_{\rm{diag}}}
\newcommand{\Toff}{T_{\rm{off}}}

\newcommand{\QF}{F_{\alpha,N}}
\newcommand{\QFp}[2]{F_{#1,#2}}
\newcommand{\QFE}{\tilde F_{\alpha,N}}
\newcommand{\QFEp}[2]{\tilde F_{#1,#2}}
\newcommand{\QFP}{F^{\per}_{\alpha,N}}
\newcommand{\QFPE}{\tilde F^{\per}_{\alpha,N}}

\newcommand{\G}{G_{\mu}}
\newcommand{\GP}{G_{\mu}^{\per}}
\newcommand{\GPp}[1]{G_{#1}^{\per}}

\newcommand{\TO}{T_{\rm{off}}^{\mu,N}}
\newcommand{\TD}{T_{\rm{dia}}^{\mu,N}}
\newcommand{\TOE}{\tilde T_{\rm{off}}^{\mu,N}}
\newcommand{\TDE}{\tilde T_{\rm{dia}}^{\mu,N}}
\newcommand{\TOP}{T_{\rm off}^{{\rm per},\mu,N}}
\newcommand{\TDP}{T_{\rm dia}^{{\rm per},\mu,N}}
\newcommand{\TOPE}{\tilde T_{\rm off}^{{\rm per},\mu,N}}
\newcommand{\TDPE}{\tilde T_{\rm dia}^{{\rm per},\mu,N}}
\newcommand{\TS}{T_{\alpha,\mu,N}}
\newcommand{\TSE}{\tilde T_{\alpha,\mu,N}}
\newcommand{\TSPE}{\tilde T_{\alpha,\mu,N}^{\rm per}}
\newcommand{\TSPEp}[2]{\tilde T_{\alpha,#1,#2}^{\rm per}}
\newcommand{\TSP}{T_{\alpha,\mu,N}^{\rm per}}

\renewcommand{\L}{L_{\mu,N}}
\newcommand{\Lp}[2]{L_{#1,#2}}
\newcommand{\LP}{L^{\rm per}_{\mu,N}}

\newcommand{\HG}{H_0^{N}}
\newcommand{\HGp}[1]{H_0^{#1}}
\newcommand{\HGF}{\hat H_0^{N}}

\newcommand{\EG}{E_{N-1}^\per}
\newcommand{\EGp}[1]{E_{#1}^\per}
\newcommand{\EGD}{E_{N}^D}
\newcommand{\EGDp}[1]{E_{#1}^D}
\newcommand{\EGDV}{E_{N}^{V,D}}

\newcommand{\lO}{L}
\newcommand{\lI}{\ell}
\newcommand{\lE}{\ell}
\newcommand{\lLT}{L}

\newcommand{\B}{B}
\newcommand{\BE}{B}
\newcommand{\BLT}{C_\lLT}

\newcommand{\dens}{\bar \rho}

\newcommand{\CB}{C_\lE}

\newcommand{\Lat}{\mathbb L}
\newcommand{\LatB}{\mathbb L^{q_1}}
\newcommand{\LatBp}[1]{\mathbb L^{#1}}
\newcommand{\LatFact}{ \lk ( \frac {2\pi}{\lE} \rk)}

\newcommand{\mcrit}{{m^\ast}}
\newcommand{\mcritMS}{{m^{\ast\ast}}}

\newcommand{\lowerBoundT}{0.58}
\newcommand{\upperBoundT}{1.73}

\newcommand{\cT}{c_T}
\newcommand{\unit}{1\!\!1}

\newcommand{\vpOut}{\vp^{p_0,\vec p_{A}}_{i,A}}
\newcommand{\hvpOut}{\hat\vp^{p_0,\vec p_{A}}_{i,A}}
\newcommand{\vpIn}{\vp^{\vec p_{A^c}}_{i,A}}
\newcommand{\hvpIn}{\hat\vp^{\vec p_{A^c}}_{i,A}}

\newcommand{\tnabla}{\tilde\nabla}

\title{Energy contribution of a point interacting impurity\\ in a Fermi gas }
\author{Thomas Moser, Robert Seiringer}
\affil{IST Austria, Am Campus 1, 3400 Klosterneuburg, Austria}
\date{July 2, 2018}
\maketitle

\begin{abstract}
We give a bound on the ground state energy of a system of $N$ non-interacting fermions in a three dimensional cubic box interacting with an impurity particle via point interactions. We show that the change in energy compared to the system in the absence of the impurity is bounded in terms of the gas density and the scattering length of the interaction, independently of $N$. Our bound holds as long as the ratio of the mass of the impurity to the one of the gas particles is larger than a critical value $\mcritMS\approx 0.36$, which is the same regime for which we recently showed stability of the system. 
\end{abstract}

\section{Introduction}

Quantum systems of particles interacting with forces of very short range allow for an idealized description in terms of point interactions. The latter are characterized by a single number, the scattering length. 
Originally point interactions were  introduced in the 1930s to model nuclear interactions \cite{Bethe1935, Bethe1935B, Fermi1936, Thomas1935, wigner1933streuung},  but  later they were also  successfully applied to many other areas of physics, like polarons (see \cite{massignan2014polarons} and references there) or cold atomic gases \cite{ZwergerBCS}. 

It was already known to Thomas \cite{Thomas1935} that the spectrum of a bosonic many-particle system depends strongly on the range of the interactions,  and that an idealized point-interacting system with more than two particles is inherently unstable, i.e., the energy is not bounded from below. 
This collapse can be counteracted by the Pauli principle for fermions with two species (e.g., spin states). In this paper we are interested in the impurity problem where there is only one particle for one of the species.

Given $N \ge 1$ fermions of one type with mass $1$ and one particle of another type with mass $m > 0$, a model of point interactions gives a meaning to the formal expression
\begin{equation}
- \frac 1 {2m} \Delta_{y}- \frac 1 2 \sum_{i=1}^N \Delta_{x_i}  + \gamma \sum_{i=1}^N \delta(x_i-y) \label{eq:pointInter} 
\end{equation}
for $\gamma \in \R$. We note that because of the antisymmetry constraint on the wavefunctions there are only interactions between particles of different species. 
The
expression \eqref{eq:pointInter} is  ill-defined in $d\geq 2$ dimensions since $H^1(\R^d)$,  the form domain of the Laplacian, contains discontinuous functions for which the meaning of the $\delta$-function as a potential is unclear.  In the following we restrict our attention to the  case $d=3$, but we note that also two-dimensional systems  exhibit interesting behavior \cite{DellAntonio1994, Dimock2004, Ulrich2017,Ulrich2018,UlrichThesis}. For $d\geq 4$ there are no point interactions as the Laplacian restricted to functions supported away from the hyperplanes of interactions is essentially self-adjoint.

A mathematically precise meaning to \eqref{eq:pointInter} in three dimensions was given in \cite{DellAntonio1994, Finco2010,Minlos2011} and we will work  with the model introduced there. 
Our analysis will start from this well-defined model,  but we note that the question whether the model can be obtained as  a limit of Schr\"odinger operators with genuine interaction potentials of shrinking support is still open. (See, however, \cite{albeverio}  for the case $N=1$, and \cite{basti} for models in one dimension.)

In this paper we study the energy contribution of the point-interacting impurity. We confine the $N+1$ particles to a box $(0,\lO)^3$
and investigate the ground state energy of the system. In particular, our goal is the show that at given mean particle density $\bar\rho = N/L^3$, the difference between the ground state energies of the interacting and the non-interacting system is bounded independently of the system size.

Previous work on this model was mostly concerned with stability and hence studied the model without confinement. 
For example, it is possible to analyze the $2+1$ model, i.e,. two fermions of one kind and one impurity of another kind, in  great detail \cite{DellAntonio1994, correggi2012stability, Correggi2015, Correggi2015B,  Becker2017, Minlos2011, Minlos2012, Minlos2014, Minlos2014B}. It turns out that the mass of the impurity plays an important role for stability. It was shown in \cite{correggi2012stability} that for the $2+1$ system there is a critical mass  $\mcrit \approx 0.0735$ such that the system is stable for $m \ge \mcrit$ and unstable otherwise. This critical mass does not depend  on the strength of the interaction, i.e., the scattering length. 

Building on these results it was shown in \cite{MoserSeiringer2017} that a similar statement holds for the $N+1$ system. In particular, it was proven that there is a critical mass $\mcritMS \approx 0.36$ such that the 
system is stable for all $m \ge \mcritMS$, independently of $N$. This bound is presumably not sharp and stability is still open for $m \in [\mcrit,\mcritMS)$. Recently also the stability of the $2+2$ system was proved in a suitable mass range 
\cite{MoserSeiringer2018}. The general case with $N+M$ particles still poses an open problem, however. 

In all cases where stability of the system was established, the ground state energy  in infinite volume is actually zero in case the scattering length is negative, and there are no bound states. For positive scattering length there are bound states, but one still expects that only a finite number of particles can bind to the impurity. In particular, the ground state energy of the $N+1$ system is bounded from below independently of $N$ \cite{MoserSeiringer2017}. Intuitively one would expect that if one confines the system to a box in order to have a non-zero mean particle density, the interaction with the impurity should again only affect a finite number of particles, and hence the {\em energy change} compared to the non-interacting system should be $O(1)$, independently of $N$. This is what we prove here. 
We note that it is sufficient to derive a lower bound on the ground state energy,  as point interactions are always attractive, i.e., they lower the energy.

Even for regular interaction potentials, it is highly non-trivial to show that an impurity causes only an $O(1)$ change to the energy of a non-interacting Fermi gas. For fixed, i.e., non-dynamical impurities, this was established in \cite{posDensitiy} as a consequence of a positive density version of the Lieb-Thirring inequality. The result in \cite{posDensitiy} applies to systems in infinite volume, as well as to systems in a box with  periodic boundary conditions.  In the appendix we provide an extension to Dirichlet boundary conditions, since this result will be an  essential ingredient in our proof.

Compared to \cite{posDensitiy} we face here two additional difficulties:  the impurity is dynamic and has a finite mass, and the interaction with the gas particles is through singular point interactions. Besides the methods of \cite{posDensitiy} and \cite{MoserSeiringer2017}, a key ingredient in our analysis is a proof of an IMS type formula for the quadratic form defining the model, which allows for a localization of the particles into regions close and far away from the impurity. It has the same form as the IMS formula for regular Schr\"odinger operators (see \cite[Thm.~3.2]{IMS}), but is much harder to prove.

\subsection{The point interaction model}
\label{sec:model}

We consider a system of $N$ fermions of mass $1$, interacting with another particle of mass $m>0$. 
Let
\begin{equation}
\HG = - \frac 1 {2m} \Lap_0 - \frac 1 2 \sum_{i=1}^N \Lap_i  \label{eq:H0}
\end{equation}
be the non-interacting part of the Hamiltonian, acting on $L^2(\R^3) \otimes L^2_{\rm as} (\R^{3N})$, where $L^2_{\rm as}$ denotes the totally antisymmetric functions in $\otimes^N L^2(\R^3)$.
The $N+1$ coordinates we denote by $x_0, x_1, \ldots, x_N \in \R^3$ and throughout this paper we will use the notation $\vec x = (x_1, \ldots, x_N)$.
If we want to exclude a set of coordinates labeled by $A \subseteq \{1, \ldots, N\}$ we use $\hat x_A = (x_i)_{i \not \in A}$ and for short $\hat x_i = \hat x_{\{i\}}$. If we want to restrict to certain coordinates we write $\vec x_A = (x_i)_{i \in A}$. 

For  $\mu > 0$, we define $\G$ as the resolvent of $\HG$ in momentum space, i.e., 
\begin{equation}\label{def:calG}
\G(k_0,\vec k) \coloneqq \left ( \frac 1{2m} k_0^2 + \frac 12 \vec k^2  + \mu \right)^{-1} \,.
\end{equation}
We denote by $\QF$ the quadratic form used in \cite{correggi2012stability,MoserSeiringer2017} describing point interactions between $N$ fermions and the impurity. Its domain is given by
\begin{equation}
D(\QF) = \left\{  \psi = \phi_\mu + \G \xi \mid \phi_\mu \in H^1(\R^3)\otimes H_{\rm as}^1(\R^{3 N}), \xi \in H^{1/2}(\R^3)\otimes H_{\rm as}^{1/2}(\R^{3 (N-1)})\right\}
\end{equation}
where $\G \xi$ is defined via its Fourier transform (denoted by a $\hat {\,\cdot\,}$) as
\begin{equation}
\widehat{\G \xi} (k_0,\vec k) =  \G(k_0,\vec k) \sum_{i=1}^N (-1)^{i+1}  \hat \xi(k_0 + k_i , \hat k_i)\,.
\label{eq:tquadraticForm} 
\end{equation}
The space $H^1_{\rm as}(\R^{3N})$ contains all totally antisymmetric functions in $H^1(\R^{3N})$. 
For a given $\psi \in D(\QF)$ and $\mu > 0$, the splitting $\psi = \phi_\mu + G_\mu \xi$ is unique. We point out that while $\phi_\mu$ depends on the choice of $\mu$,  $\xi$ is  independent  of $\mu$. We will call $\phi_\mu$ the regular part and $\xi$ the singular part of $\psi$. Note that $D(F_{\alpha,N})$ is independent of the choice of $\mu$, and so is the quadratic form $F_{\alpha,N}$ defined as
\begin{align}
\QF(\psi) &\coloneqq \left\langle \phi_\mu \left| \HG  + \mu \right| \phi_\mu \right\rangle  - \mu \norm{\psi}^2_{L^2(\R^{3(N+1)})} + \TS(\xi) \\ 
\TS(\xi) & \coloneqq N  \left( \frac {2m }{m+1} \alpha \norm{\xi}_{L^2(\R^{3N})}^2 +   \TD(\xi)+  \TO(\xi) \right)  \label{def:cFal}
\end{align}
where
\begin{align}
\TD(\xi) & \coloneqq \int_{\R^{3N}}  |\hat \xi(\vec k)|^2  \L(\vec k) \dI \vec k  \\
\TO(\xi) & \coloneqq (N-1) \int_{\R^{3(N+1)}} \hat \xi^\ast (k_0+ k_1 ,\hat k_1) \hat \xi (k_0 + k_2 , \hat k_2) G_\mu(k_0,\vec k)  \dI k_0 \dI \vec k \label{def:calT}\\
\L(\vec k) & \coloneqq 2 \pi^2 \left(\frac{2m}{m+1}\right)^{3/2}  \lk ( \frac {k_1^2}{2(m+1)} + \frac 12 \hat k_1^2 + \mu \rk)^{1/2} \,.
\label{def:L}
\end{align}
The quadratic form $F_{\alpha,N}$ describes $N$ fermions interacting with an impurity  particle via point interactions with scattering length $a = -2\pi^2/\alpha$, with $\alpha\in \R$. The non-interacting system is recovered in the limit $\alpha\to +\infty$. 

{\em Notation.} 
Throughout the paper we will use the following notation.
We define the relation $\lesssim$ by
\begin{equation}
x \lesssim y \iff \exists C>0 \colon x \le C y 
\end{equation}
where $C$ is independent of $x$ and $y$. In the obvious way we define $\gtrsim$. In case that $x \lesssim y$ and $ y \lesssim x$ we write $x \sim y$.

\section{Main result for confined wavefunctions}
\label{sec:Results}
Let us assume that $\supp \psi \subseteq \B^{N+1}$, where $\B = (0,\lO)^3$ for some $L>0$. The mean particle density will be denoted by  $\dens = N/\lO^3$. Let $\EGD$ be the ground state energy of $-\frac 12 \sum_{i=1}^N  \Lap_i$ for wavefunctions in $H^1_{\rm as}(\R^{3N})$ with Dirichlet boundary conditions on $\partial\B$. It equals the sum of the lowest $N$ eigenvalues of the Dirichlet Laplacian on $B$, and it is easy to see that
\begin{equation}
\EGD \sim N \dens^{2/3}\,. 
\end{equation}
A natural question is how the interactions affect this energy. From \cite[Thm. 2.1]{MoserSeiringer2017} we know that there is a mass-dependent constant $\Lambda(m)$ \cite[Eq.~(2.8)]{MoserSeiringer2017}, given in Eq.~\eqref{eq:BigLambda} below, such that if $\Lambda(m) < 1$ then $\QF$ is bounded from below independently of $N$ by
\begin{equation}
\frac {\QF(\psi)}{\norm{\psi}^2_2} \ge \frac {m+1}{2m} \begin{cases} 0 & \alpha \ge 0 \\ - \lk ( \dfrac {\alpha}{2\pi^2 (1- \Lambda(m))} \rk )^2 & \text{otherwise.} \end{cases} \label{eq:bound3} 
\end{equation}
(The additional factor $(m+1)/(2m)$ compared to \cite[Thm.~2.1]{MoserSeiringer2017} results from the separation of the center-of-mass motion used in \cite{MoserSeiringer2017}.) 
It was also shown in \cite{MoserSeiringer2017} that $\Lambda(m) < 1$ if $m > \mcritMS \approx 0.36$.

For particles confined to the box $B$ with mean density $\dens$ 
we can show that under the condition $\Lambda(m) < 1$ the correction  to $\EGD$ is small, i.e., it is $O(1)$ independently of $N$. Our main result is the following.

\begin{theorem}
\label{thm:main}
Let $\psi \in D(\QF)$, supported in $(0,\lO)^{3(N+1)}$, with $\norm{\psi} = 1$. Let $\dens = N \lO^{-3}$, and assume that $\Lambda(m)<1$. Then \begin{equation}
\QF(\psi) \ge  \EGD - \const \lk  ( \dfrac {\dens^{2/3}} {(1- \Lambda(m))^{9/2}}  + \dfrac {\alpha^2_-}{(1 - \Lambda(m))^{2}} \rk ) \label{eq:boundMain}    
\end{equation}
where the constant is independent of $\psi,m, N, \lO$ and $\alpha$, and $\alpha_-$  denotes the negative part of $\alpha$, i.e., $\alpha_-=\frac 12(|\alpha|-\alpha)$. 
\end{theorem}

Thm.~\ref{thm:main} shows that the presence of the impurity affects the ground state energy  by a term that is bounded independently of $N$. The bound \eqref{eq:boundMain} is an extension of \eqref{eq:bound3} in the sense that if we take $\lO \to \infty$ in \eqref{eq:boundMain} we recover \eqref{eq:bound3} up to the value of the constant.

\begin{remark*}
For $\alpha \to \infty$ one would expect that the optimal lower bound converges to the  ground state energy of the non-interacting Hamiltonian $\HG$ with Dirichlet boundary conditions. This is not the case for \eqref{eq:boundMain} which is independent of $\alpha$ for $\alpha \ge 0$. 
\end{remark*}

Using various types of trial states the ground state energy of point-interacting systems is extensively discussed in the physics literature (see \cite{massignan2014polarons} and references there). We note that with this method it is only possible to derive upper bounds, while  Thm.~\ref{thm:main} gives a lower bound on the ground state energy.

\subsection{Proof outline}
For the proof of Theorem \ref{thm:main} we first prove in Section~\ref{sec:IMS}  an IMS type formula, which allows to localize the impurity in a small box, of side length $\ell$ independent of $L$. In a second step we localize all of the remaining particles to be either close to the impurity or separated from it. Doing this we partly violate the antisymmetry constraint on the wavefunctions, which makes it necessary to first extend the quadratic form $\QF$ to $\QFE$. The latter does not require the antisymmetry, but coincides with $\QF$ on $D(\QF)$.

In Section \ref{sec:Apriori} we give a rough lower bound on the energy in  case the wavefunction is compactly supported in a box $(0,\lE)^3$. This lower bound is of the order $N^{5/3}/\ell^2$, as expected,  but with a non-sharp prefactor. We shall introduce a quadratic form $\QFP$ with periodic boundary conditions and show that it is equivalent to $\QF$ for  confined wavefunctions. The reason we work with periodic boundary conditions instead of Dirichlet ones is that it allows to  perform explicit computations in momentum space. 

Because the ground state energy of the confined non-interacting $N$-particle system is strictly positive, we are allowed to choose  $\mu$ negative in the definition of $\QFP$. Applying the method of \cite{MoserSeiringer2017} then leads to the lower bound on $\QFP$ in Theorem \ref{thm:AprioriBound}.
The downside of working with $\QFP$ will be that because of the discrete nature of momentum space for periodic functions, we have to work with sums instead of integrals, and the difference between  the sum and the integral versions will have to be carefully controlled.  

In Section \ref{sec:Loc} we give the proof of Theorem~\ref{thm:main}. Using the IMS formula of Prop.~\ref{thm:IMS}, we localize the particles  either in a small box with side length  $\lI \sim \dens^{-1/3}$ containing the impurity, or in the large complement. In the small  box we use Theorem \ref{thm:AprioriBound} for a lower bound,  whereas in the large complement we use Theorem~\ref{thm:lti}, which is a version of the positive density Lieb-Thirring inequality in \cite{posDensitiy} adapted to our setting of  Dirichlet boundary conditions, and which is proved in the appendix.  This allows us to improve the rough bound of Thm.~\ref{thm:AprioriBound} and show Thm.~\ref{thm:main}.

\section{Properties of the quadratic form}
\label{sec:IMS}
In this section we will first extend the quadratic form $\QF$ to functions that are not required to be antisymmetric in the last $N$ variables. Afterwards we shall  discuss how the splitting $\psi= \phi_\mu + \G \xi$ is affected  when multiplying $\psi$ by a smooth function (which need not be symmetric under permutations). This will be utilized in the last part of this section where an IMS formula for the (extended) quadratic form is shown. 

\subsection{Extension to functions without symmetry}
To prove our main theorem, we want to localize the particles in different subsets of the cube $B=(0,L)^3$. Hence it is necessary to extend the quadratic form $\QF$ by removing the antisymmetry constraint. To this aim we define
\begin{equation}
D(\QFE) = \left\{  \psi = \phi_\mu + \sum_{i=1}^N G_\mu \xi_i \mid \phi_\mu \in H^1(\R^{3 (N+1)}), \xi_i \in H^{1/2}(\R^{3 N}) \ \forall \,i, 1 \le i \le N\right\}
\end{equation}
where 
\begin{equation}
\widehat{\G \xi_i} (k_0,\vec k) =  \G(k_0, \vec k) \hat \xi_i(k_0 + k_i ,\hat k_i)\,.
\label{eq:tquadraticFormNoSym} 
\end{equation}
The quadratic form $\QFE$ is defined as
\begin{align}
\QFE(\psi) & \coloneqq \left\langle \phi_\mu \left| \HG   + \mu \right| \phi_\mu \right\rangle  - \mu  \norm{\psi}^2_{L^2(\R^{3(N+1)})} + \TSE(\vec \xi)  \\ \TSE(\vec \xi)  &\coloneqq   \frac {2m }{m+1} \alpha \sum_{i=1}^N \norm{\xi_i}_{L^2(\R^{3N})}^2 +  \TDE(\vec \xi) + \TOE(\vec \xi) \label{eq:quadraticFormNoSym}
\end{align}
where $\vec \xi = (\xi_i)_{i=1}^N$ and
\begin{align}
\TDE(\vec \xi) & \coloneqq \sum_{i=1}^N \int_{\R^{3N}}  |\hat \xi_i(\vec k)|^2  \L(\vec k)  \dI \vec k  \\
\TOE(\vec \xi) & \coloneqq - \sum_{\substack{i \ne j\\ 1 \le i,j \le N}} \int_{\R^{3(N+1)}} \hat \xi^\ast_i (k_0+ k_i , \hat k_i) \hat \xi_j (k_0 + k_j , \hat k_j) G_\mu(k_0,\vec k)  \dI k_0 \dI \vec k \,.
\label{def:calTNoSym}
\end{align}
Each $\xi_i$ in \eqref{eq:tquadraticFormNoSym} corresponds to a function supported on the hyperplane $x_0=x_i$. The only overlap between hyperplanes for  $i\neq j$ is on the set $x_i = x_0 = x_j$, which 
implies that $\sum_{i=1}^N \hat \xi_i(k_0+k_i, \hat k_i)$ has a unique decomposition into $(\xi_i)_{i=1}^N$, and thus the splitting $\psi = \phi_\mu + \sum_{i=1}^N \G \xi_i$ is  unique. To stress the dependence on $\psi$, we will sometimes use the notation $\phi^\psi_\mu$ and $\xi_i^\psi$ below.

In the case that $\psi$ is antisymmetric in the last $N$ coordinates, the uniqueness of the decomposition $\psi = \phi_\mu + \sum_{i=1}^N \G \xi_i$ shows that there exists a function $\xi \in  H^{1/2}(\R^3)\otimes H_{\rm as}^{1/2}(\R^{3 (N-1)})$ such that
$\xi_i = (-1)^{i+1} \xi$, and hence $\sum_{i=1}^N \G \xi_i = \G \xi$, defined in \eqref{eq:tquadraticForm}. 
Furthermore we have
\begin{equation}
\TDE(\vec \xi) = N \TD(\xi), \qquad \TOE(\vec \xi) =  N \TO(\xi) \label{eq:conExtNonExtT} 
\end{equation}
in this case, which shows that $\QFE(\psi) = \QF(\psi)$ for $\psi$ antisymmetric in the last $N$ coordinates. In particular, $\QFE$ is an extension of $\QF$, and for a lower bound it therefore suffices  to work with $\QFE$.

In the following, it will be convenient to introduce the notation
\begin{equation}
\tnabla \coloneqq \lk ( \frac 1 {\sqrt{2m}} \nabla_0, \frac 1 {\sqrt 2}  \nabla_1, \ldots, \frac 1 {\sqrt 2} \nabla_N \rk ) \label{eq:nabla} 
\end{equation}
as well as
\begin{equation}\label{def:hmu}
H_\mu \coloneqq H_0^N + \mu = -\tnabla^2 + \mu \,.
\end{equation}

\subsection{Localization of wavefunctions}
An important ingredient in the proof of Theorem~\ref{thm:main} 
will be to localize the particles. For this purpose we will study in this subsection how the splitting $\psi = \phi^\psi_\mu + \sum_{i=1}^N \G \xi_i^\psi$ is affected when multiplying $\psi$ by a smooth function. 

\begin{lemma}
\label{lem:Jtrans}
For $J \in \CC^\infty(\R^{3(N+1)})$ bounded and with bounded derivatives, we define $J \vec \xi = (J \xi_i)_{i=1}^N$ by
\begin{equation}\label{def:jacts}
(J \xi_i)(x_i, \hat x_i) = J(x_i,\vec x) \xi_i(x_i, \hat x_i) \, .
\end{equation}
Then $\xi_i \mapsto [J,G_\mu] \xi_i \coloneqq J G_\mu \xi_i - G_\mu J \xi_i $ is a bounded map from $L^{2}(\R^{3N})$ to $H^1(\R^{3(N+1)})$. In particular
\begin{equation}
\xi_i^{J \psi} = J \xi_i^\psi \label{eq:Jmap} 
\end{equation}
and the regular part $\phi^{J \psi}_\mu$ of $J\psi$ is given by 
\begin{equation}
\phi^{J \psi}_\mu = J \phi^{\psi}_\mu + \sum_{i=1}^N [J,\G] \xi_i^\psi \, .  \label{eq:Phimap} 
\end{equation}
\end{lemma}

\begin{remark*}
We clarify that $J$ acts on functions on $\R^{3(N+1)}$, and  in particular on $\phi_\mu^\psi$ and  $G_\mu \xi_i^\psi$, as a multiplication operator, whereas on functions in $L^2(\R^{3N})$ it acts as in \eqref{def:jacts}. Hence the commutator $[J, \G]$ has no meaning here independently of its application on $\vec \xi$, and is only used as a convenient notation. 
\end{remark*}

\begin{proof}
We first  argue that $[J,G_\mu] \xi_i^\psi \in H^1(\R^{3(N+1)})$ implies \eqref{eq:Jmap} and \eqref{eq:Phimap}.
We have
\begin{equation}
J \psi - \sum_{i=1}^N \G J \xi_i^\psi 
= J \phi_\mu^\psi + \sum_{i=1}^N [J,\G] \xi_i^\psi \,.\label{eq:JTerms} 
\end{equation}
Since $J \phi_\mu^\psi$ and  $ [J,G_\mu]\xi_i^\psi $ are in $H^1(\R^{3(N+1)})$, the uniqueness of the decomposition of $J\psi$ into regular and singular parts implies  \eqref{eq:Jmap} and \eqref{eq:Phimap}.

It remains to show that $[J,G_\mu] \xi_i \in H^1(\R^{3(N+1)})$ for $\xi_i \in L^2(\R^{3N}$. In order to do so, we shall in fact show that
\begin{equation}\label{ne}
  [J,\G] \xi_i = H_\mu^{-1}  [ H_0^N , J ]\G \xi_i =  H_\mu^{-1}  ( - 2 \tnabla\cdot(\tnabla J) -(\tnabla^2 J)  ) \G \xi_i \,,
\end{equation}
where we used the notation introduced in \eqref{eq:nabla} and \eqref{def:hmu}. 
From \eqref{ne} the $H^1$ property readily follows, using that 
\begin{equation}
\norm{\G \xi_i}^2_{L^2(\R^{3(N+1)})}  = \int_{\R^{3(N+1)}} \G(k_0,\vec k)^2 |\hat\xi_i(k_0+k_i, \hat k_i)|^2 \dI k_0 \dI \vec k \lesssim \left(\frac{m}{m+1}\right)^{3/2} \mu^{-1/2} \norm{\xi_i}^2_{L^2(\R^{3N})}\,.  \label{3:14}
\end{equation}
In the last step we did an explicit integration over $\frac 1{m+1}k_0- \frac m{m+1}k_i$, the variable canonically conjugate to $x_0-x_i$. 

In order to show \eqref{ne}, we note that since $J$ is smooth, $H_\mu^{-1} J H_\mu$ is a bounded operator. In the sense of distributions, we have
\begin{equation}
\left( H_\mu \G \xi_i \right)(x_0,\vec x) = \xi_i(  x_i, \hat x_i) \delta(x_0-x_i)
\end{equation}
and hence $H_\mu^{-1} J H_\mu \G\xi_i = \G J \xi_i$. In particular, 
\begin{equation}\label{usee}
  [J,\G] \xi_i  =  \left( J - H_\mu^{-1} J H_\mu  \right) \G \xi_i
\end{equation}
which indeed equals \eqref{ne}. 
This completes the proof of the lemma.
\end{proof}

\begin{corollary}
\label{cor:domainXi}
Assume that $\psi \in D(\QFE)$ satisfies  $\supp \psi \subseteq \Omega_0 \times \cdots \times \Omega_N$, where $\Omega_j \subseteq \R^3$ for $0 \le j \le N$.
Then 
\begin{equation}
\supp \xi_i^\psi \subseteq (\Omega_0 \cap \Omega_i) \times \Omega_1 \times \cdots \times \Omega_{i-1} \times \Omega_{i+1} \times \cdots \times \Omega_N \,.
\end{equation}
\end{corollary}

\begin{proof}
Let $J \in \CC^\infty(\R^{3(N+1)})$ such that $J(x_0,\vec x) = 1$ for $(x_0,\vec x) \in \Omega_0 \times \cdots \times \Omega_N$. Using  Lemma~\ref{lem:Jtrans} we get that
\begin{equation}
\xi_i^{\psi}(x_0, \hat x_i) = \xi^{J \psi}_i(x_0, \hat x_i) = J(x_i,\vec x) \xi^\psi_i(x_i, \hat x_i) \,.
\end{equation}
Since this holds for all $J$ with the above property, the claim follows.
\end{proof}

\subsection{Alternative representation of the singular part}

The following Lemma gives an alternative representation of the singular part of the quadratic form, defined in \eqref{eq:quadraticFormNoSym}. It will turn out to be  useful in the proof of the IMS formula in the next subsection.

\begin{lemma}
For  $\vec \xi = (\xi_i)_{i=1}^N$ with  $\xi_i \in H^{1/2}(\R^{3 N})$, the function 
\begin{equation}
\mathcal{I}(\nu) \coloneqq \left\| \sum\nolimits_{i=1}^N G_\nu \xi_i \right\|_{L^2(\R^{3(N+1)})}^2 -  \pi^2 \left(\frac{2m}{m+1}\right)^{3/2} \frac{1 }{ \sqrt{ \nu}} \sum_{i=1}^N \norm{\xi_i}_{L^2(\R^{3N})}^2
\end{equation}
is integrable on $[\mu,\infty)$ for any $\mu>0$, and 
we have
\begin{equation}\label{rep:sing}
\TSE(\vec \xi)  = \left( \frac {2m }{m+1} \alpha   + 2 \pi^2 \left(\frac{2m}{m+1}\right)^{3/2}    \sqrt\mu \right)  \sum_{i=1}^N \norm{\xi_i}_{L^2(\R^{3N})}^2  - \int_{\mu}^\infty \dI \nu \,  \mathcal{I}(\nu ) \,.
 \end{equation}
\end{lemma}

\begin{proof}
For any $1\leq i\leq N$, we have
\begin{align}\nonumber
\norm{G_\nu \xi_i}^2_{L^2(\R^{3(N+1)})}  & = \int_{\R^{3(N+1)}} G_\nu (k_0,\vec k)^2 |\hat\xi_i(k_0+k_i, \hat k_i)|^2 \dI k_0 \dI \vec k = 
\\ & =  \left(\frac{2m}{m+1}\right)^{3/2}  \int_{\R^{3N}}  \frac{ \pi^2} {\sqrt{ \tfrac {k_i^2}{2(1+m)} + \tfrac 12 \hat k_i^2 + \nu }}  |\hat\xi_i(k_i, \hat k_i)|^2  \dI k_0 \dI \vec k \,.
\end{align}
In particular, 
\begin{equation}
 \left\| G_\nu \xi_i \right\|_{L^2(\R^{3(N+1)})}^2 -  \left(\frac{2m}{m+1}\right)^{3/2} \frac{ \pi^2 }{ \sqrt\nu} \norm{\xi_i}_{L^2(\R^{3N})}^2  \leq 0
\end{equation}
and we have 
\begin{align}\nonumber
& -\int_\mu^\infty \dI \nu \left(  \left\| G_\nu \xi_i \right\|_{L^2(\R^{3(N+1)})}^2 -  \left(\frac{2m}{m+1}\right)^{3/2} \frac{ \pi^2 }{ \sqrt\nu} \norm{\xi_i}_{L^2(\R^{3N})}^2  \right) \\ & = \int_{\R^{3N}}  |\hat \xi_i(\vec k)|^2  \L(\vec k)  \dI \vec k -  2 \pi^2 \left(\frac{2m}{m+1}\right)^{3/2}    \sqrt \mu   \norm{\xi_i}_{L^2(\R^{3N})}^2 \,.
\end{align}
For the terms $i\neq j$, on the other hand, we have
\begin{align}\nonumber
\int_\mu^\infty \dI \nu \, \langle G_\nu \xi_i | G_\nu \xi_j\rangle & =  \int_\mu^\infty \dI \nu 
\int_{\R^{3(N+1)}} \hat \xi^\ast_i (k_0+ k_i , \hat k_i) \hat \xi_j (k_0 + k_j , \hat k_j) G_\nu(k_0,\vec k)^2 \dI k_0 \dI \vec k \\
& = \int_{\R^{3(N+1)}} \hat \xi^\ast_i (k_0+ k_i , \hat k_i) \hat \xi_j (k_0 + k_j , \hat k_j) G_\mu(k_0,\vec k) \dI k_0 \dI \vec k \,.
\end{align}
Here the exchange of the order of integration is justified by Fubini's theorem, since the integrand in the first line on the right is absolutely integrable for $\xi_i \in H^{1/2}$. This completes the proof.  
\end{proof}

\subsection{IMS formula}
In this subsection we will prove the following Lemma.

\begin{prop}
\label{thm:IMS}
Given $M \ge 1$ and $(J_i)_{i=1}^{M}$ with $J_i \in \CC^\infty(\R^{3(N+1)})$ and $\sum_{i=1}^M J_i^2 = 1$, we have 
\begin{equation}
\QFE(\psi) = \sum_{i=1}^M  \QFE(J_i \psi) - \sum_{i=1}^M \norm{ (\tnabla J_i) \psi}^2  
\end{equation}
for all $\psi\in D(\tilde F_{\alpha,N})$.
\end{prop}

\begin{proof}
By using the polarization identity, we can extend $\QFE$ to a  sesquilinear form, denoted as $\QFE(\psi_1,\psi_2)$.  It suffices to prove that
\begin{equation}
\QFE(J^2 \psi,  \psi) + \QFE(\psi, J^2 \psi) - 2 \QFE (J \psi, J \psi) = - 2\norm{ (\tnabla J) \psi }^2 \label{eq:statement} 
\end{equation}
for smooth functions $J$, since then 
\begin{equation}
\QFE(\psi) = \frac 1 2 \sum_{i=1}^M \lk ( \QFE(J_i^2 \psi, \psi) + \QFE(\psi, J_i^2 \psi) \rk ) \overset{\substack{\eqref{eq:statement}}} = \sum_{i=1}^M \QFE(J_i \psi, J_i \psi) - \sum_{i=1}^M  \norm{ (\tnabla J_i) \psi}^2\,.  
\end{equation}

Recall the definition $H_\mu = H_0^N + \mu$. The left side of \eqref{eq:statement} equals
\begin{align}\nonumber
& \braket{ \phi^{J^2 \psi}_\mu}{ H_\mu| \phi_\mu^\psi} + \braket {\phi_\mu^\psi} { H_\mu| \phi^{J^2 \psi}_\mu} - 2 \braket {\phi^{J \psi}_\mu} { H_\mu| \phi^{J\psi}_\mu} \\ &  + \TSE(\vec \xi^{J^2 \psi}, \vec \xi^\psi) + \TSE(\vec \xi^\psi, \vec \xi^{J^2 \psi}) - 2 \TSE( \vec \xi^{J \psi}, \vec \xi^{J \psi}) \label{2lines}
 \end{align}
where we introduced the sesquilinear form $\TSE(\vec\xi_1,\vec \xi_2)$ corresponding to the quadratic form \eqref{eq:quadraticFormNoSym}. 
We use Lemma~\ref{lem:Jtrans} to identify the regular and singular parts of the various wavefunctions. 
For the quadratic form $\TSE$, we utilize the representation \eqref{rep:sing}, which together with \eqref{eq:Jmap}  implies that
\begin{align}\nonumber
&  \TSE(\vec \xi^{J^2 \psi}, \vec \xi^\psi) + \TSE(\vec \xi^\psi, \vec \xi^{J^2 \psi}) - 2 \TSE( \vec \xi^{J \psi}, \vec \xi^{J \psi}) \\
& = \int_\mu^\infty \dI \nu \sum_{i,j=1}^N \left( 2 \langle G_\nu J \xi_i^\psi | G_\nu J \xi_j^\psi \rangle -  \langle G_\nu J^2 \xi^\psi_i | G_\nu  \xi_j^\psi \rangle - \langle G_\nu  \xi_i^\psi | G_\nu J^2 \xi_j^\psi \rangle\right) \,. \label{ipf1}
 \end{align}
Since $G_\nu J \xi_i^\psi = H_\nu^{-1} J H_\nu G_\nu \xi^\psi$, as shown in the proof of Lemma~\ref{lem:Jtrans}, we can rewrite the terms in the integrand as
\begin{align}\nonumber
&2 \langle G_\nu J \xi_i^\psi | G_\nu J \xi_j^\psi \rangle -  \langle G_\nu J^2 \xi_i^\psi | G_\nu  \xi_j^\psi \rangle - \langle G_\nu  \xi_i^\psi | G_\nu J^2 \xi_j^\psi \rangle \\  & = 
 \left \langle G_\nu  \xi_i^\psi \left| 2 H_\nu J  H_\nu^{-2} J H_\nu  -   H_\nu^{-1} J^2  H_\nu  - H_\nu J^2  H_\nu^{-1}  \right|   G_\nu \xi_j^\psi \right\rangle \,. \label{tfe}
\end{align}
Using that $(\partial/\partial\nu) G_\nu \xi_i^\psi = - H_\nu^{-1} G_\nu \xi_i^\psi$ as well as $[J,[H_\nu,J]] = 2 |\tnabla J|^2$, one readily checks that this further equals
\begin{equation}\label{ipf3}
\eqref{tfe} = -2 \frac\partial{\partial\nu}  \left \langle G_\nu  \xi_i^\psi \left| [J,H_\nu] H_\nu^{-1} [H_\nu,J]  - |\tnabla J|^2 \right|   G_\nu \xi_j^\psi \right\rangle \,.
\end{equation}
The operator $A_\nu\coloneqq [J,H_\nu] H_\nu^{-1} [H_\nu,J]  - |\tnabla J|^2$ is bounded, uniformly in $\nu$ for $\nu\geq \mu>0$. Since $\|G_\nu \xi_i^\psi\|_2 \to 0$ as $\nu\to \infty$, we have $\lim_{\nu\to\infty} \langle G_\nu \xi_i^\psi| A_\nu | G_\nu \xi_j^\psi\rangle =0$. 
In particular, from \eqref{ipf1}--\eqref{ipf3} we conclude that
\begin{align}\nonumber
&\TSE(\vec \xi^{J^2 \psi}, \vec \xi^\psi) + \TSE(\vec \xi^\psi, \vec \xi^{J^2 \psi}) - 2 \TSE( \vec \xi^{J \psi}, \vec \xi^{J \psi})
\\ &  = \sum_{i,j=1}^N \lk (  
 2 \left \langle \G  \xi_i^\psi \left| [J,H_\mu] H_\mu^{-1} [H_\mu,J]   \right|   \G \xi_j^\psi \right\rangle
 - 2 \braket{\G \xi_i^{\psi}}{ |\tnabla J|^2 \G \xi_j^{\psi}} \rk)\, . \label{eqli3}
\end{align}

For the regular part, we use \eqref{eq:Phimap}  to rewrite the first line in \eqref{2lines} as 
\begin{align}\nonumber
& \braket{ \phi^{J^2 \psi}_\mu}{ H_\mu| \phi_\mu^\psi} + \braket {\phi_\mu^\psi} { H_\mu| \phi^{J^2 \psi}_\mu} - 2 \braket {\phi^{J \psi}_\mu} { H_\mu| \phi^{J\psi}_\mu} \\  \nonumber & =  -2 \braket{ \phi^{\psi}_\mu}{ |\tnabla J|^2  \phi_\mu^\psi}  - 2  \sum_{i,j=1}^N  \braket{ [J, \G ] \xi_i^\psi }{ H_\mu | [J,\G] \xi_j^\psi}  \\ & \quad - 4 \Re  \sum_{i=1}^N \braket{ [ J, \G] \xi_i^{\psi} } {H_\mu |  J \phi_\mu^\psi} + 2\Re \sum_{i=1}^N  \braket{ [J^2, \G] \xi_i^{\psi} } {H_\mu |  \phi_\mu^\psi}\,.  \label{nee}
\end{align}
The second term on the right side equals $ -2 \sum_{i,j=1}^N  \left \langle \G  \xi_i^\psi \left| [J,H_\mu] H_\mu^{-1} [H_\mu,J]   \right|   \G \xi_j^\psi \right\rangle$, as \eqref{ne} shows.  
Also the last line in \eqref{nee} can be evaluated with the aid of \eqref{ne}, with the result that 
\begin{align}\nonumber
&- 4 \Re  \sum_{i=1}^N \braket{ [ J, \G] \xi_i^{\psi} } {H_\mu |  J \phi_\mu^\psi} + 2\Re \sum_{i=1}^N  \braket{ [J^2, \G] \xi_i^{\psi} } {H_\mu |  \phi_\mu^\psi} 
\\ & = - 4 \Re  \sum_{i=1}^N \braket{ \G \xi_i^{\psi} } { |\tnabla J|^2  \phi_\mu^\psi} \,. \label{eqli4}
\end{align}
In combination, \eqref{eqli3}, \eqref{nee} and \eqref{eqli4} imply the desired identity \eqref{eq:statement}.
This completes the proof of the lemma.
\end{proof}

\section{A rough bound}
\label{sec:Apriori}

In this section we give a rough lower bound on the ground state energy of $\QF$ when restricted to wavefunctions $\psi \in D(\QF)$ that are supported in  $\BE^{N+1}$ with $\BE = (0,\lE)^3$ for some $\ell>0$. This lower bound has the desired scaling in $N$ and $\lE$, i.e., it is proportional to $N^{5/3} \lE^{-2}$, but with a non-sharp prefactor.  For its proof, we will first reformulate the problem using periodic boundary conditions, and then apply the  methods previously introduced in \cite{MoserSeiringer2017} to show stability in infinite space.

The statement of the following theorem involves three positive constants $\cT$, $c_L$ and $c_\Lambda$, which are independent of $m,N,\ell$ and $\alpha$ and which will be defined later.  In particular, $\cT$ is defined in Eq.~\eqref{def:cT}, $c_L$ in Eq.~\eqref{fdef:cL} and $c_\Lambda$ in Lemma~\ref{lem:cL}. 

\begin{theorem}
\label{thm:AprioriBound}
Let $\psi \in D(\QF)$ with $\|\psi\|=1$ and $\supp \psi \subseteq (0,\lE)^{3(N+1)}$ or some $\lE> 0$. Given $m>0$ and $\kappa>0$ such that 
\begin{equation}
1-  \kappa/  \cT > \Lambda(m) \label{eq:condKappa} 
\end{equation}
let $N_0 = N_0(m,\kappa)$ be defined as 
\begin{equation}\label{def:n0}
N_0(m,\kappa) = \lk ( \lk(1 -  \kappa/\cT - \Lambda(m) \rk)  \frac{ m (1- \kappa/\cT)^2 }{ c_\Lambda} \rk)^{-9/2} \,.
\end{equation}
For $N > N_0$ we have
\begin{equation}\label{lbthm41}
\QF(\psi) \ge 
\kappa N^{5/3} \lE^{-2} - \frac 1 {4 \pi^4} \frac {m+1} {2 m}   \dfrac { [\alpha- c_L  \ell^{-1}]^2_-}{(1-\kappa/\cT - \Lambda(m))^2 ( 1 - (N_0/N)^{2/9} )^2 } \,.
\end{equation}
\end{theorem}

We note that this result gives a lower bound only for particle numbers  $N > N_0(m,\kappa)$.
In the case that $N \leq N_0$, we can still use \eqref{eq:bound3}, however. 

The remainder of this section contains the proof of Theorem \ref{thm:AprioriBound}. 
An important role will be played by a reformulation using  periodic boundary conditions. 
We will start by introducing the functional $\QFPE$ which is defined for periodic functions. 
In Lemma \ref{lem:periodic} we will show that it is in fact equivalent to the original quadratic form $\QFE$   when applied to wavefunctions with compact support in $\BE^{N+1}$. 
Working with periodic boundary conditions comes with the inconvenience of having to work with sums, rather than with integrals, in momentum space. In particular, this makes the explicit form of the singular part  of $\QFPE$ rather complicated; we shall compare it with  the singular part of $\QFE$ in  Lemma \ref{lem:AP1} and bound the difference. It comes with the big advantage of allowing us to choose  $\mu$ negative, however, which will be essential to show a positive lower bound to the  energy. 
We shall use the method of \cite{MoserSeiringer2017} which gives positivity of the singular part of $\QFP$ for $\mu \geq -\kappa N^{5/3} \ell^{-2}$ for small enough $\kappa$, under a condition of the form $\tilde \Lambda(m,\kappa)<1$. In  Lemmas~\ref{lem:pointwiseBound}--\ref{cor:DiffLambda}, we investigate the difference between $\tilde \Lambda(m,\kappa)$ and $\Lambda(m)$. In the last subsection we combine these results to prove Theorem \ref{thm:AprioriBound}.

\subsection{Periodic boundary conditions} 
Given $\psi \in D(\QFE)$ such that $\supp \psi \subseteq B^{N+1}$, we extend  $\psi$ to a periodic function $\psi^\per$, defined as 
\begin{equation}
\psi^{\rm{per}}(x_0, \ldots, x_{N}) = \psi(\tau(x_0), \ldots, \tau(x_{N}))
\end{equation}
with
\begin{equation}
\tau(x) = (\tau(x^1), \tau(x^2),\tau(x^3)),\qquad  \tau(s) \coloneqq \inf \lk ( \lk(s + \lE \mathbb Z \rk) \cap \R_+ \rk) \ \ \text{ for } s \in \R. 
\end{equation}
In the following we shall rewrite the functional $\QFE(\psi)$ in terms of $\psi^{\rm per}$.  Compared to Dirichlet boundary conditions, periodic ones have the advantage that one can work easily in the associated momentum space, similar to the unconfined case.  For this purpose, we define 
the lattice in momentum space as
\begin{equation}
\Lat \coloneqq \frac {2 \pi} {\lE} \mathbb Z^3 \,. \label{def:Lat} 
\end{equation}
The function $\psi^{\rm per}$ is then determined by its Fourier coefficients $\hat \psi^{\rm per}(k_0,\vec k)$, which can be viewed as a  function $\Lat^{N+1} \to \C$.

Corollary \ref{cor:domainXi} implies that  $\supp \xi_i \subseteq \BE^N$ for all $1\leq i\leq N$. Hence we can extend it in a similar way as $\psi$ to a periodic function $\xi^\per$. 
In momentum space we can write it as $\hat\xi^\per: \Lat^N \to \C$. 
For periodic functions, $\G \psi^\per$ does not make sense anymore, but instead choosing $\GP$ as the resolvent of the  non-interacting Hamiltonian with periodic boundary conditions allows us to define $\GP \xi^\per_i$ by the Fourier coefficients
\begin{equation}
\widehat {\GP \xi_i^\per}(k_0,\vec k) = \G(k_0, \vec k) \hat\xi_i^\per(k_0+k_i, \hat k_i) \label{eq:GPer} \,.
\end{equation}

In order to motivate the quadratic form introduced below, we note that the expression $\L(\vec k)$ in \eqref{def:L}  originates from the limit
\begin{equation}
\L(\vec k) = \lim_{R \to \infty} \lk ( \frac {8  \pi m R}{m+1} - \int_{|t|\leq R} \frac 1 {\tilde H_0(k_1,t, \hat k_1) + \mu} \dI t  \rk) \label{def:LLimit} 
\end{equation}
where $\tilde H_0$ is the non-interacting Hamiltonian in momentum space, expressed in terms of center-of-mass and relative coordinates for the pair $(k_0,k_1)$, i.e., 
\begin{equation}
\tilde H_0(s ,t,\hat k_1) \coloneqq \hat H_0^N  \lk (\frac m {m+1} s + t,\frac {1}{m+1}s -t,\hat k_1 \rk ) 
=  \frac 1 {2 (m+1)} s^2 + \frac {1 + m}{2 m} t^2 + \frac 1 2 \hat k_1^2\, . \label{def:Hshift}
\end{equation}
More generally, we have

\begin{lemma}\label{def:tau}
Let $\tau$ be a non-negative function in $\CC^\infty_0(\R^3)$ such that $\hat \tau(0) = 1, \hat \tau(p) \ge 0$ for all $p \in \R^3$ and
\begin{equation}
\int_{\R^3} |t|^{-2} \tau(t) \dI t  = 4 \pi \,.  \label{eq:propSeq} 
\end{equation}
Then
\begin{equation}
\L(\vec k) = \lim_{R \to \infty}  \lk [  \frac {8 \pi mR }{m+1} - \int_{\R^3} \frac {1} { \tilde H_0(k_1 ,t,\hat k_1) + \mu } \hat \tau(t/R) \dI t  \rk ]  \,. \label{eq:indCutoff} 
\end{equation}
\end{lemma}

\begin{proof}
Let $\gamma = \frac 1 {2(m+1)} k_1^2 + \frac 1 2 \hat k_1^2 + \mu$. 
Using  \eqref{eq:propSeq} we observe that 
   \eqref{eq:indCutoff} is equivalent to 
\begin{equation}
\lim_{R\to \infty}  \int_{\R^3} \frac {\gamma}{\lk ( \lk ( \frac {1+m}{2m} \rk) t^2 + \gamma\rk )\lk ( \frac {1+m}{2m} \rk) t^2 } \hat \tau(t/R) \dI t  = \L(\vec k)\,. 
\end{equation}
Since $\hat \tau(0) = 1$ and $\hat \tau(t) \le 1$ for all other $t$, the result follows from dominated convergence. 
\end{proof}

When replacing integrals by sums, we have to keep in mind that  a change of coordinates from $(k_0,k_1)$ to $s = k_0 + k_1$ and $t = \frac {m}{m+1} k_1 - \frac 1 {m+1} k_0$ changes the domain over which we have to take the sums. Whereas $s \in \Lat$ we have to sum for a fixed $s$ the variable $t$ over $\Lat^s \coloneqq \Lat + \frac {m s}{m+1}$.
Let $\tau$ be chosen as in Lemma~\ref{def:tau}, and define
\begin{equation}\label{def:LP}
\LP(\vec k) \coloneqq \lim_{R \to \infty} \lk ( \frac {8 \pi m R}{m+1} - \lk(\frac {2 \pi} {\lE} \rk )^3 \sum_{\substack{p \in \Lat^{k_1}}} \frac 1 {\tilde H_0(k_1, p ,\hat k_1) + \mu} \hat \tau(p/R)   \rk ) \,.
\end{equation}
We shall see below that this definition is actually independent of $\tau$. 
For us it will be important that $\tau$ has compact support, hence a sharp cut-off in momentum space would not be suitable.

We shall now define $\QFPE$ with domain 
\begin{equation}
D(\QFPE) = \left\{  \psi^\per = \phi_\mu^\per + \sum_{i=1}^N \G^\per \xi_i^\per \mid \phi_\mu^\per \in  H_{\per}^1(\BE^{N+1}),  \, \xi_i^\per \in H^{1/2}_\per(\BE^N)\ \forall \,i, 1 \le i \le N\right\}\,,
\end{equation}
where $H^1_\per(\BE^{N+1})$ and $H^{1/2}_\per(\BE^{N})$ denotes the spaces  of functions defined by Fourier  coefficients  in $\ell^2(\Lat,(1+p^2))^{\otimes (N+1)}$ and  $\ell^{2}(\Lat,(1+p^2)^{1/2})^{\otimes N}$ respectively. The quadratic form is given by
\begin{align}
\QFPE(\psi^{\rm per}) & \coloneqq \int_{\BE^{N+1}}  \lk (  |\tnabla \phi_\mu^\per|^2    + \mu  |\phi_\mu^\per|^2 \right )  - \mu \norm{\psi^\per}^2_{L^2(\BE^{N+1})} + \TSPE(\vec \xi^\per) \\  \TSPE(\vec \xi^\per) &\coloneqq  \sum_{i=1}^N  \frac {2m }{m+1} \alpha \norm{\xi_i^\per}_{L^2(\BE^N)}^2 + \TDPE(\vec \xi^\per)+\TOPE(\vec \xi^\per) \label{eq:quadraticFormNoSymPer}
\end{align}
where $\vec \xi^\per = (\xi_i^\per)_{i=1}^N$, $\tnabla$ is defined in \eqref{eq:nabla}, and the
 singular parts of the quadratic form are given by
\begin{align}
\TDPE(\vec \xi^\per) & \coloneqq \sum_{i=1}^N \lk ( \frac {2\pi}{\lE^3} \rk)^{3N} \sum_{\vec k \in \Lat^N}  |\hat \xi_i^\per(\vec k)|^2  \LP(\vec k) \label{def:calTNoSymPer}   \\
\TOPE(\vec \xi^\per) & \coloneqq - \sum_{\substack{i \ne j\\ 1 \le i,j \le N}} \lk ( \frac {2\pi}{\lE^3} \rk)^{3(N+1)} \sum_{k_0 \in \Lat, \vec k \in \Lat^{N}}  {{} \hat\xi^\per_j}^\ast (k_0+ k_j , \hat k_j) \hat \xi_i^\per (k_0 + k_i, \hat k_i) \G(k_0,\vec k) \,. \label{eq:TPer} 
\end{align}

We also define $\QFP$ as the restriction of $\QFPE$ to functions antisymmetric in the last $N$ coordinates. Further we define $\TDP, \TOP$ and $\TSP$ in the natural way similar to $\TD, \TO$ and $\TS$ originating from $\TDE, \TOE$ and $\TSE$, respectively (compare with \eqref{def:cFal} and \eqref{eq:conExtNonExtT}).

\begin{lemma}
\label{lem:periodic}
Let $\psi \in D(\QFE)$ be such that $\supp \psi \subseteq \BE^{N+1}$. Then
\begin{equation}
\QFPE(\psi^\per) = \QFE(\psi) \,.
\end{equation}
\end{lemma}

\begin{proof}
Recall the splitting of $\psi$ into its regular and singular parts, and similarly for $\psi^{\rm per}$: 
\begin{equation}
\psi = \phi_\mu + \sum_i \G \xi_i \ , \quad 
\psi^\per  = \phi_\mu^\per + \sum_i \GP \xi_i^\per \,.
\end{equation}
Recall also the definition \eqref{def:hmu}. 
In the sense of distributions we can apply $H_\mu$ to $\phi_\mu$, and in particular $H_\mu \phi_\mu \in H^{-1}(\R^{3(N+1)})$ as $\phi_\mu \in H^1(\R^{3(N+1)})$. In this sense we can write the regular part of $\QFE$ as $\braket{\phi_\mu}{ H_\mu \phi_\mu}$. 
Because $\supp \psi \subseteq B^{N+1}$ we have $\ve \coloneqq {\rm dist}(\supp \psi, \pd B) > 0$. Let $\chi$ be a smooth cutoff function such that $\chi(x) = 1$ if $x \in B_0 = [\ve/2, \lE-\ve/2]^3$ and $\chi(x) = 0$ if $x \in \BE^c$. As $\supp (H_\mu \G \xi) \subseteq B_0^{N+1}$ and $\supp \psi \subseteq B_0^{N+1}$ also $\supp (H_\mu \phi_\mu) \subseteq B_0^{N+1}$, and therefore
\begin{equation}
\braket{\phi_\mu}{H_\mu\phi_\mu} = \braket{\chi \phi_\mu}{ H_\mu \phi_\mu} \,.  \label{eq:PNP1} 
\end{equation}
We use the identity $\chi \phi_\mu = \chi \phi_\mu^\per + \chi \sum_{i=1}^N \GP \xi^\per_i -  \chi \sum_{i=1}^N \G \xi_i$ as well as the fact  that $H_\mu \phi_\mu = H_\mu \phi_\mu^\per$ on $B_0^{N+1}$ to obtain
\begin{align}\nonumber 
\eqref{eq:PNP1} &= \braket{\chi \phi_\mu^\per}{H_\mu\phi_\mu^\per}  + \sum_{i=1}^N  \braket{\chi (\GP \xi^\per_i - \G \xi_i)}{H_\mu\phi_\mu^\per} \\
&=  \int_{\BE^{N+1}}  \lk ( |\tnabla \phi_\mu^\per|^2  + \mu  |\phi_\mu^\per|^2 \right )  + \sum_{i=1}^N  \braket{\chi (\GP \xi^\per_i - \G \xi_i)}{H_\mu \phi_\mu^\per} \,.\label{eq:PNP2} 
\end{align}
Note that $H_\mu \chi ( \G^\per \xi_i^\per - \G \xi_i)$ is supported on $ \BE \setminus B_0$, and $\psi^\per$ vanishes on this set. Hence  
\begin{equation}
\sum_{i=1}^N \braket{\chi (\GP \xi^\per_i - \G \xi_i)}{H_\mu \phi_\mu^\per} = 
-\sum_{i,j=1}^N \braket{\GP \xi^\per_i - \G \xi_i}{\chi H_\mu \GP \xi^\per_j} \label{eq:PerNoPer2} \,.
\end{equation}

We claim that \eqref{eq:PerNoPer2} is equal to the difference $\TSPE(\vec \xi^\per) - \TSE(\vec \xi)$. Let $\tau$ be given as in Lemma~\ref{def:tau}. We approximate the distribution  $(\chi H_\mu \GP \xi^\per_j)(x_0, \vec x) = \xi_j( x_j ,\hat x_j) \delta(x_j- x_0)$ by the sequence of functions $(\xi_j \tau_R)(x_0, \vec x) = \xi_j( (m x_j + x_0)/(1+m),\hat x_j) \tau_R(x_j- x_0)$ with $\tau_R(x) = R^{3} \tau(R x)$. We assume that $R$ is large enough such that $ \tau_R $ is supported in a ball of radius $\ve/2$, and hence $\xi_j \tau_R$ is supported in $B^{N+1}$.  
Because $\GP \xi^\per_i - \G \xi_i$ is actually a smooth function, as $H_\mu (\GP \xi^\per_i - \G \xi_i) = 0$ on $\BE^{N+1}$, we conclude that \eqref{eq:PerNoPer2} is equal to
\begin{equation}
\eqref{eq:PerNoPer2} = - \lim_{R \to \infty} \sum_{i,j=1}^N  \braket{\GP \xi^\per_i - \G \xi_i}{\xi_j \tau_R}  \,. \label{eq:PerNoPer3} 
\end{equation}
For the terms with $i\neq j$, we can use dominated convergence in momentum space to conclude that  
\begin{equation}
\lim_{R \to \infty} \sum_{i\neq j}  \braket{\GP \xi^\per_i - \G \xi_i}{\xi_j \tau_R} = \TOE(\vec \xi)  - \TOPE(\vec \xi^\per) \,.
\end{equation}
For the terms with $i=j$, we can further write
\begin{align}\nonumber
&  \sum_{i=1}^N  \braket{\GP \xi^\per_i - \G \xi_i}{\xi_i\tau_R} \\
&=   \sum_{i=1}^N \left( \braket{\GP \xi^\per_i}{\xi_i\tau_R} - \frac {8 \pi m R}{m+1} \norm{\xi_i}^2_2\right)  -  \sum_{i=1}^N \left( \braket{\G \xi_i}{\xi_i \tau_R} - \frac {8 \pi m R}{m+1} \norm{\xi_i}^2_2 \right)  \,. 
\end{align}
Lemma~\ref{def:tau} implies that the limit of the last two terms exists, is independent of the choice of $\tau$ and is equal to $\TDE(\vec \xi)$. Because also \eqref{eq:PerNoPer2} does not depend on $\tau$ we  conclude that
\begin{equation}
\lim_{R \to \infty} \sum_{i=1}^N  \lk (  \braket{\GP \xi_i^\per}{\xi_i \tau_R} - \frac {8 \pi m R}{m+1}\norm{\xi_i}^2_2 \rk) 
\end{equation}
exists and is independent of $\tau$. Comparing with \eqref {def:LP} and \eqref{def:calTNoSymPer}, we see that it actually equals $\TDPE(\vec \xi^{\rm per})$.  Combining the above, we obtain 
\begin{equation}
\braket{\phi_\mu}{H_\mu \phi_\mu} =  \int_{\BE^{N+1}}  \lk ( |\tnabla \phi_\mu^\per|^2  + \mu  |\phi_\mu^\per|^2 \right )  + \TSPE(\vec \xi^\per) - \TSE(\vec \xi)\,. 
\end{equation}
This completes the proof of the lemma.
\end{proof}

For fermions, described by wavefunctions $\psi^{\rm per}$ that are antisymmetric in the last $N$ variables, 
the expression $\GP \xi^\per$ in \eqref{eq:GPer} is also well defined for negative $\mu$ as long as $\mu > - \EG$, where $\EG$ denotes the ground state energy of the non-interacting Hamiltonian for $N-1$ fermions with periodic boundary conditions on $\partial B$. (Note  than $\G \xi$, on the other hand,  is only defined for $\mu > 0$.)  
The following lemma shows that for such $\mu$ the quadratic form $\QFP$ is actually  independent of $\mu$.

\begin{lemma}
\label{cor:nmu}
For $\psi \in D(\QFP)$, the expression $\QFP(\psi^\per)$ is well-defined and independent of $\mu$ as long as $\mu > - \EG$.
\end{lemma}

\begin{proof}
 We first note that $\GP \xi^\per$ is well defined for $\mu > -\EG$, because of the antisymmetry of $\xi^{\per}$ in the last $N-1$ variables, which implies that $N-1$ of the variables  $(k_1,\dots,k_N)$ in $\G(k_0,\vec k)$ in \eqref{eq:GPer} are actually different. 
For $\nu,\mu > -\EG$ we have
\begin{equation}
\phi^\per_{\mu} = \phi^\per_\nu + \GPp \nu \xi^\per - \GP \xi^\per \,. 
\end{equation}
Using the resolvent identity, we see that the regular part of the quadratic form satisfies  
\begin{align}\nonumber 
\int_{\BE^{N+1}}  \lk (  |\tnabla \phi_\mu^\per|^2    + \mu  |\phi_\mu^\per|^2 \right )  &= \int_{\BE^{N+1}}  \lk (  |\tnabla \phi_\nu^\per|^2    + \nu  |\phi_\nu^\per|^2 \right )  
 +  (\mu - \nu) \norm{\phi^\per_\nu}^2   \\
&\quad + 2 (\mu-\nu) \Re  \braket{ \GPp \nu \xi^\per}{ \phi^\per_\nu} + (\mu-\nu) \braket{ \GPp \nu \xi^{\rm per} }{\GPp \nu \xi^\per - \GP \xi^\per}  \,.
\end{align}
A straightforward computation using the definitions \eqref{def:LP}--\eqref{eq:quadraticFormNoSymPer} shows that 
\begin{equation}
 \TSP( \xi^\per) -  T^{\rm per}_{\alpha,\nu,N}( \xi^\per) = (\mu -\nu) \langle G_\nu^{\rm per} \xi^{\rm per} | \G^{\rm per} \xi^{\rm per} \rangle \, . 
\end{equation}
Combining both statements yields the desired identity
\begin{align}\nonumber
& \int_{\BE^{N+1}}  \lk (  |\tnabla \phi_\mu^\per|^2    + \mu  |\phi_\mu^\per|^2 \right )   - \mu \norm{\psi^{\rm per}}^2 + \TSP( \xi^\per) \\ & = \int_{\BE^{N+1}}  \lk (  |\tnabla \phi_\nu^\per|^2    + \nu  |\phi_\nu^\per|^2 \right )  
 - \nu \norm{\psi^{\rm per}}^2 +  T^{\rm per}_{\alpha,\nu,N}( \xi^\per) \, . 
\end{align}
\end{proof}

\subsection{Approximation by integrals}
In the previous subsection we have shown that the original and the  periodic  formulations of the energy functionals, $\QFE$ and $\QFPE$, agree if applied to functions $\psi$ compactly supported in $\BE^{N+1}$. One complication in the periodic form is that $\LP$ is not given as explicitly as $\L$. The following lemma gives a bound on the difference.

\begin{lemma}
\label{lem:AP1}
Given $\mu$ and $\vec q$ such that 
\begin{equation}
Q_\mu^2 \coloneqq \frac 1 2 \sum_{i=2}^N q_i^2 + \mu >0 \label{eq:AP1Cond} 
\end{equation} 
we have 
\begin{equation}
|\LP(q_1,\hat q_1) -  \L(q_1,\hat q_1)| \leq c_L'    \frac 1{Q_\mu^2 \lE^{3}} 
\end{equation}
where the constant $c_L'$ is independent of $N,\vec q,m,\lE$ and $\mu$.
\end{lemma}

\begin{proof}
We recall the definitions of $\L$ and $\LP$ for some arbitrary $\tau$ fulfilling the requirements of Lemma~\ref{def:tau}:
\begin{align}\nonumber 
\L(\vec q) &= - \lim_{R \to \infty} \lk ( \int \frac {1}{\tilde H_0(q_1,s, \hat q_1) + \mu} \hat \tau(s/R) \dI s - \frac {8 \pi m R}{m+1} \rk) \\
\LP(\vec q) &= - \lim_{R \to \infty} \lk ( \LatFact^3 \sum_{s \in \Lat^{q_1}} \frac {1}{\tilde H_0(q_1,s, \hat q_1) + \mu} \hat \tau(s/R) - \frac {8 \pi m R}{m+1} \rk) 
\end{align}
with $\tilde H_0$ defined in \eqref{def:Hshift}.
For simplicity we assume that $q_1$ is such that $\Lat^{q_1} = \Lat$ but all other cases work analogously as a shift in momentum space only introduces a phase factor in configuration space, which vanishes when taking absolute values.
In the following we denote $f_\infty(s) = (\tilde H_0(q_1, s ,\hat q_1) + \mu)^{-1}$ and $f_R(s) = f_\infty(s) \hat \tau(s/R)$ and suppress the dependence on $\vec q$ for simplicity.

We can express the difference between the Riemann sum and the integral using Poisson's summation formula
\begin{equation}
\LatFact^3 \sum_{s \in \Lat} f_R(s) - \int_{\R^3} f_R(s) \dI s = \frac {(2\pi)^3} {\lE^3} \sum_{s \in \Lat} f_R( s) - (2 \pi)^{3/2} \hat f_R(0) =  (2\pi)^{3/2} \sum_{\substack{z \in \lE \mathbb Z^3 \\ z \ne 0}} \hat f_R(z)\, . \label{eq:diffPoisson} 
\end{equation}
For short we write $\gamma \coloneqq  \frac 1 {2 (1+m)} q_1^2 + \frac 1 2 \hat q_1^2+\mu$, which is bounded from below by $Q_\mu^2$ and hence is positive,  by our assumption \eqref{eq:AP1Cond}. The function $f_\infty$ and its Fourier transform are given by
\begin{equation}
f_\infty(t)  = \frac 1 {\frac{1 +m}{2m} t^2 + \gamma}, \quad \hat f_\infty(z) = \sqrt{\frac \pi 2} \frac {2m}{1+m} \frac { e^{- \lk ( \frac {2m}{m+1} \rk)^{1/2} \sqrt {\gamma} |z|} } {|z|} \, . 
\end{equation}
Moreover,  
\begin{equation}
\hat f_R(z) = (2 \pi)^{-3/2} (R^3 \tau(R \, \cdot \, ) \ast \hat f_\infty)(z) \,.
\end{equation}
We will show that $\hat f_R(s)$ is summable over $\lE \mathbb Z^3 \setminus \{0\}$. In fact for $|z| \gtrsim \ell$,
\begin{align}\nonumber
(2 \pi)^{3/2} |\hat f_R(z)| &=   \int_{\R^3} R^3 \tau(R w) \hat f_\infty(z-w) \dI w  \\ \nonumber 
 &\le   \int_{|w| > |z|/2} R^3 \tau(R w) \hat f_\infty(z - w) \dI w +   \int_{|z-w| > |z|/2} R^3 \tau (R w) \hat f_\infty(z -w) \dI w \\
 &\le   \hat f_\infty(z/2) \int R^3 \tau(R w) \dI w  =  \hat f_\infty(z/2)
 \label{eq:HelperFunctionConvergence} 
\end{align}
where we assumed that  $R$ is large enough such that $\tau(R w) = 0$ for $|w| > |z|/2$, and used that $\int \tau  = 1$, which was required by Lemma~\ref{def:tau}.
As $\hat f_\infty$ is summable over $\lE \mathbb Z^3 \setminus \{0\}$ we get by dominated convergence that
\begin{equation}
\lim_{R \to \infty} \sum_{z \in \lE \mathbb Z^3 \setminus \{0\}} |\hat f_R(z)| = \sum_{z \in \lE \mathbb Z^3 \setminus \{0\}} \hat f_\infty(z)\,. \label{eq:HelperFunctionConvergence2}  
\end{equation}
We bound the sum over $\hat f_\infty(|z|)$ by
\begin{equation}
\sum_{z \in \lE \mathbb Z^3 \setminus \{0\}} \hat f_\infty(z) =  \sum_{n \in \mathbb Z^3 \setminus \{0\}} \sqrt{\frac \pi 2} \frac {2m}{1+m} \frac { e^{- \lk ( \frac {2m}{m+1} \rk)^{1/2} \sqrt {\gamma} \lE |n|} } {\lE |n|} \lesssim  \frac 1 {\gamma \lE^3}   \label{eq:sumfinf} 
\end{equation}
using
\begin{equation}
\sum_{n \in \mathbb Z^3 \setminus \{0\} } e^{-\eta |n|}/|n| \lesssim \sum_{n \in \mathbb N} n e^{- \eta n} = \frac {e^{- \eta}} {(1 - e^{- \eta})^2} \le \frac 1 {\eta^2} 
\end{equation}
for $\eta = (2m/(m+1))^{1/2} \sqrt{\gamma} \lE$. Combining \eqref{eq:diffPoisson}, \eqref{eq:HelperFunctionConvergence2} and \eqref{eq:sumfinf} and using that $\gamma\geq Q_\mu^2$, we conclude that 
\begin{equation}
\lim_{R\to \infty} \lk | \frac {(2\pi)^3} {\lE^3} \sum_{s \in \Lat} f_R(s) - \int_{\R^3} f_R(s) \dI s \rk| \le  \frac {c_L'} {\gamma \lE^3} \le \frac {c_L'}{Q_\mu^2 \ell^3}
\end{equation}
for some constant $c_L' > 0$. 
This completes the proof of the lemma.
\end{proof}

\subsection{Bound on the singular parts}

The strategy for obtaining a lower bound on $\QFP$ is to find a $\mu$ such that $\TSP \geq 0$, in which case we obtain the lower bound $\QFP(\psi^{\rm per}) \geq -\mu \|\psi^{\rm per}\|^2$. Hence we want to choose $\mu$ as negative as possible. We shall 
use the method  of \cite{MoserSeiringer2017}, which yields the desired positivity of $\TSP$ (for large enough $m$) as long as 
 $\mu \ge -\kappa N^{5/3} \lE^{-2}$ for $\kappa$ small enough. (More precisely, $-\mu$ will be equal to the right side of \eqref{lbthm41}.)
 
 If we define $Q^2 = \frac 1 2 \sum_{i=2}^N q_i^2$ for $N > 2$, we observe that there exists a constant $\cT>0$ such that
\begin{equation}
Q^2 \ge \cT N^{5/3} \lE^{-2} \label{def:cT}
\end{equation}
if all $q_i \in \Lat$ are different, as required by the antisymmetry constraint. (We note that in comparison with \cite{MoserSeiringer2017}  $Q^2$ is defined with an additional factor $1/2$ here.) 
From now on we  restrict $\mu$ to satisfy  $\mu \ge -\kappa N^{5/3} \lE^{-2}$ for some  $\kappa < \cT$. 
This implies that
\begin{equation}
Q_\mu^2 =  Q^2 + \mu \ge (1- \kappa/ \cT ) Q^2 \ge (\cT -  \kappa)  N^{5/3} \lE^{-2}  \,. \label{eq:condMuQ}  
\end{equation}
In particular, Lemma~\ref{lem:AP1} yields the bound 
\begin{equation}\label{cons44}
\TDP(\xi^\per) \geq   \LatFact^{3N}  \sum_{\vec q \in \Lat^{N}} \L(\vec q)|\hat \xi^{\rm per}(\vec q)|^2 - \frac 1 {N^{5/3} \ell} \frac {c_L'}{\cT-\kappa}   \| \xi^\per\|_2^2
\end{equation} 
on the diagonal term of the singular part of $\QFP$. 
Following the same steps as in \cite{MoserSeiringer2017} we can obtain the following lower bound for the off-diagonal term.
\begin{prop}
\label{lem:ToffEst}
Assume that  $\mu \ge -\kappa N^{5/3} \lE^{-2}$ for some $\kappa < c_T$. Then for all $\xi \in H^{1/2}(\R^{3})\otimes H_{\rm as}^{1/2}(\R^{3(N-1)})$ we have
\begin{equation}
\TOP(\xi^{\rm per}) \geq - \frac{ \tilde \Lambda(m,\kappa)}{1- \kappa/\cT}  \LatFact^{3N}  \sum_{\vec q \in \Lat^{N}} \L(\vec q)|\hat \xi^{\rm per}(\vec q)|^2  
\end{equation}
where
\begin{equation} 
\tilde \Lambda(m,\kappa) \coloneqq \inf_{\delta > 0} \sup_{\substack{\tilde s,K\in\R^3\\ Q_\mu^2 > (\cT  - \kappa) N^{5/3} \lE^{-2} }} \LatFact^{3} \sum_{\tilde t \in \Lat +A K} \lambda_{\tilde s,Q_\mu,K,m,\delta}(\tilde t) \label{eq:BigLambdaTilde} 
\end{equation}
with
\begin{align}\nonumber
\lambda_{\tilde s,Q_\mu,K,m,\delta}(\tilde t) &  \coloneqq    \frac{(\tilde s-AK)^2 +  2 Q_\mu^2 + N \delta \lE^{-2}}{\pi^2 (1+m)} \left( \frac{m(m+2)}{(m+1)^2} \tilde s^2 + \frac m{m+1} ( 2 Q_\mu^2 + A K^2)  \right)^{-1/4} \\ \nonumber & \qquad  \times  \frac 1{(\tilde t-AK)^2+\delta \lE^{-2}}     \left( \frac{m(m+2)}{(m+1)^2} \tilde t^2 + \frac m{m+1} (2 Q_\mu^2+AK^2)  \right)^{-1/4}   \\ &  \qquad \times   \frac { \left|\tilde s\cdot \tilde t\right|  }{ \left[ \tilde s^2 + \tilde t^2 + \frac m{1+m}(2 Q_\mu^2 +A K^2)  \right]^2  - \left[ \frac 2{(1+m)} \tilde s\cdot \tilde t\right]^2} \,.  \label{def:lambda}
\end{align}
\end{prop}

\begin{proof}
The proof works in almost the exact same way as in \cite{MoserSeiringer2017}, hence we will not spell out the details. The main difference is that we now have to write sums instead of integrals,  and in particular this implies that we have to choose the weight function $h(s,\hat q_1)$ (see \cite[Eq. (4.12)]{MoserSeiringer2017}) differently, namely as
\begin{equation}
h(s,\hat q_1) = (s^2 + \delta \lE^{-2}) \prod_{i=2}^N (q_i^2 + \delta \lE^{-2}) \,. \label{eq:weighing} 
\end{equation}
For comparison $\delta = 0$ was used in \cite{MoserSeiringer2017}. Following the proof in \cite[Sect.~4]{MoserSeiringer2017} this choice gives a lower bound to the off-diagonal term of the form 
\begin{equation}
\TOP(\xi^{\rm per}) \geq  - \tilde \Lambda_{\delta,\mu}(m) \LatFact^{3N}  \sum_{\vec q \in \Lat^N} \L(\vec q)|\hat \xi^{\rm per}(\vec q)|^2 \label{eq:TrelLambda} 
\end{equation}
with a prefactor $\tilde \Lambda_{\delta,\mu}(m)$ equal to 
\begin{align} \nonumber
 &   \sup_{\tilde s,K\in\R^3, Q^2 >  \cT N^{5/3} \lE^{-2}}  \frac{(\tilde s-AK)^2 + 2 Q^2 + N \delta \lE^{-2}}{\pi^2 (1+m)} \left( \frac{m(m+2)}{(m+1)^2} \tilde s^2 + \frac m{m+1} ( 2 Q^2 + A K^2) + \frac {2m} {m+1} \mu \right)^{-1/4} \\ \nonumber & \qquad \qquad  \times  \LatFact^{3} \sum_{\tilde t \in \Lat + A K}  \frac 1{(\tilde t-AK)^2+\delta \lE^{-2}}     \left( \frac{m(m+2)}{(m+1)^2} \tilde t^2 + \frac m{m+1} (2 Q^2+AK^2) + \frac {2m} {m+1} \mu  \right)^{-1/4}   \\ & \qquad\qquad \times   \frac { \left|\tilde s\cdot \tilde t\right|  }{ \left[ \tilde s^2 + \tilde t^2 + \frac {m}{1+m}( 2 Q^2 +A K^2) + \frac {2 m}{m+1} \mu \right]^2  - \left[ \frac 2{(1+m)} \tilde s\cdot \tilde t\right]^2}\,.  \label{eq:LambdaTilde}
\end{align}
Since \eqref{eq:condMuQ} holds under our assumption on $\mu$, we see that $\inf_{\delta>0}  \tilde \Lambda_{\delta,\mu}(m) \leq (1 - \kappa/ \cT)^{-1} \tilde \Lambda(m,\kappa)$, which yields the desired result.
\end{proof}

\subsection{A bound on $\tilde \Lambda(m,\kappa)$}

We will not evaluate $\tilde \Lambda(m,\kappa)$ directly but we will compare it with $\Lambda(m)$,  which is  defined in \cite[Eq. (2.8)]{MoserSeiringer2017} and which was already referred to  in \eqref{eq:bound3} above. The expression $\Lambda(m)$ can be written as
\begin{equation} 
\Lambda(m)  \coloneqq  \sup_{\substack{\tilde s,K\in\R^3\\ Q_\mu^2 > 0}} \int_{\R^3} \lambda_{\tilde s,Q_\mu,K,m,0}(\tilde t) \dI \tilde t
= \sup_{\substack{\tilde s,K\in\R^3\\ Q_\mu^2 > (\cT - \kappa)  N^{5/3} \lE^{-2} }} \int_{\R^3} \lambda_{\tilde s,Q_\mu,K,m,0}(\tilde t) \dI \tilde t \,.\label{eq:BigLambda} 
\end{equation}
The additional constraint on $Q_\mu$ in the latter supremum has no effect because of  the scaling properties of $\lambda_{\tilde s, Q_\mu, K,m,0}$, specifically  
$\lambda_{\nu \tilde s, \nu Q_\mu, \nu K,m,0}(\nu \tilde t) =  \nu^{-3} \lambda_{\tilde s, Q_\mu, K,m,0}(\tilde t)$ for any $\nu>0$, which allows to fix one of the parameters when taking the supremum.  
The expression \eqref{eq:BigLambdaTilde} differs from \eqref{eq:BigLambda}  by the non-zero value of $\delta$, as well as the sum instead of an integral. In the following lemmas we will compare the  two.

The next Lemma gives a pointwise bound on $\lambda_{\tilde s,Q_\mu,K,m,\delta}$. For its statement it will be
convenient to define $\CB(s)$ as the cube with side length $2 \pi/\lE$ centered at $s \in \R^3$, i.e., 
\begin{equation}
\CB(s) = \lk [-\frac {\pi}{\lE},\frac {\pi}{\lE} \rk]^3 + s. \label{def:CB} 
\end{equation} 

\begin{lemma} 
\label{lem:pointwiseBound}
For $m\gtrsim 1$ we have
\begin{equation}
\lambda_{\tilde s,Q_\mu,K,m,\delta}(\tilde t) \lesssim \frac 1 {m}  \frac 1 {t^{5/2}} \frac{s^2 + 2 Q_\mu^2 + N\delta \ell^{-2}}{(s^2 + 2 Q_\mu^2 )^{1/4}} \frac 1 {s^2 + t^2 +2 Q_\mu^2}  \label{eq:pointWise} 
\end{equation}
where $\tilde s = s + AK$ and $\tilde t = t + AK$ for $t \in \Lat \setminus \{0\}$. Moreover, 
\begin{equation}\label{eq:sum}
\ell^{-3} \sum_{\tilde t \in \Lat + AK} \max_{\tau \in \CB(\tilde t)} \lambda_{\tilde s,Q_\mu,K,m,\delta}(\tau) \lesssim \frac 1m  \lk ( 1 +  \frac {N \delta}{ \lE^{2} Q_\mu^2 } + \frac {1}{\delta \ell Q_\mu } + \frac {N}{ \lE^3 Q_\mu^3} \rk)\,.
\end{equation}
\end{lemma}

\begin{proof}
For the pointwise bound \eqref{eq:pointWise} we will proceed similarly to \cite[Sect.~6]{MoserSeiringer2017}. Using the Cauchy-Schwarz inequality we have
\begin{equation}\label{4.59}
|\tilde t \cdot \tilde s| \le  \frac 1 2  \lk [\tilde s^2 + \tilde t^2 + \frac {m}{1+m}( 2 Q_\mu^2 +A K^2)  \rk]
\end{equation}
and also
\begin{equation}
\left[ \tilde s^2 + \tilde t^2 + \frac {m}{1+m}( 2 Q_\mu^2 +A K^2)  \right]^2  - \left[ \frac 2{(1+m)} \tilde s\cdot \tilde t\right]^2  \ge \frac {m(m+2)}{(1+m)^2} \left[ \tilde s^2 + \tilde t^2 + \frac {m}{1+m}( 2 Q_\mu^2 +A K^2)  \right]^2\,.
\end{equation}
By minimizing over $K$ we find that 
\begin{equation}
\tilde s^2 + \tilde t^2 + \frac {m}{1+m}( 2 Q_\mu^2 +A K^2)  \ge \frac {m (2+m)}{2+4m+m^2} \lk [ s^2 + t^2 + 2 Q_\mu^2 \rk] 
\end{equation}
and 
\begin{equation}
 \frac{m(m+2)}{(m+1)^2} \tilde s^2 + \frac m{m+1} ( 2 Q_\mu^2 + A K^2) \ge \frac {m}{m+1} \lk(s^2 + 2 Q_\mu^2  \rk) \,.
\end{equation}
By combining these bounds we get for \eqref{def:lambda} the pointwise bound
\begin{align}\nonumber
\lambda_{\tilde s,Q_\mu, K,m,\delta}(\tilde t) & \le   \left(\frac{m+1}{m}\right)^{3/2} \frac{m^2+4 m+2}{ 2\pi^2 m  (m+2)^2}  \lk (s^2 +  2 Q_\mu^2 + N \delta  \lE^{-2} \rk )   \\ &  \quad  \times \left( s^2 + 2Q_\mu^2   \right)^{-1/4} \frac 1{ t^2+\delta \lE^{-2}}     \left( t^2 + 2 Q_\mu^2  \right)^{-1/4}    \frac { 1 }{ s^2 +  t^2 + 2 Q_\mu^2 }  \label{eq:pointWiseDet} 
\end{align}
from which \eqref{eq:pointWise} readily follows.

We denote the right side of \eqref{eq:pointWise} by $\lambda^>(t) = \lambda^>_{ s,Q_\mu,K,m,\delta}(t)$ and we will write $\lambda(\tilde t)= \lambda_{\tilde s,Q_\mu,K,m,\delta}(\tilde t)$ in the following. That is, \eqref{eq:pointWise} reads $\lambda(\tilde t)\lesssim \lambda^>(t)$. 
First we treat the term $\tilde t=AK$ in \eqref{eq:sum}. Using \eqref{eq:pointWiseDet} 
we can bound 
\begin{equation}
\lE^{-3} \lambda(\tilde t) \lesssim \frac {1}{ m\delta  \lE Q_\mu} \frac{s^2 +  2 Q_\mu^2 + N \delta  \lE^{-2}}{s^2 +  t^2 + 2 Q_\mu^2 }  \lesssim \frac 1m\lk ( \frac 1 {\delta \ell Q_\mu } + \frac {N}{\lE^3Q_\mu^3 } \rk )  
\end{equation}
for any $\tilde t$ and hence, in particular, for $\tilde t \in C_\ell(AK)$. 
For the case $0\neq t\in \Lat$, we note that for $\tau_1,\tau_2 \in C_\ell( t)$ the bound $|\tau_1|\leq \sqrt{11} |\tau_2|$ holds, and hence 
\begin{equation}
 \lambda^>(\tau_1) \le 11^{9/4} \lambda^>(\tau_2) \,.
\end{equation}
In particular,  the maximal value of $\lambda^>$ in $C_\ell(\tau)$ is dominated by the average value, 
and therefore
\begin{align}\nonumber 
\ell^{-3} \sum_{\tilde t \in \Lat + AK} \max_{\tau \in \CB(\tilde t)} \lambda(\tau) &\lesssim \ell^{-3} \sum_{\substack{t \in \Lat\\ t \ne 0}}  \lambda^>(t) +  \frac 1m \lk ( \frac 1 {\delta \ell Q_\mu } + \frac {N}{\lE^3Q_\mu^3 } \rk ) \\ \nonumber
& \lesssim \sum_{\substack{t \in \Lat\\t \ne 0}} \int_{\CB(t)}  \lambda^>( t) \dI t + \frac 1  m \lk ( \frac 1 {\delta \ell Q_\mu } + \frac {N}{\lE^3Q_\mu^3 } \rk )  \\
& \lesssim \int_{\R^3} \lambda^>(t) \dI t + \frac 1m \lk ( \frac 1 {\delta \ell Q_\mu } + \frac {N}{\lE^3Q_\mu^3 } \rk )   \,.
\end{align}
As a last step we  explicitly evaluate the integral, which results in the bound
\begin{equation}
\int_{\R^3} \lambda^>(t) \dI t \lesssim  \frac 1m \lk ( 1 +  \frac {N \delta}{\lE^2 Q_\mu^2  } \rk) \,.
\end{equation}
This completes the proof of the lemma.
\end{proof}

\newcommand{\RN}[1]{%
  \textup{\uppercase\expandafter{\romannumeral#1}}%
}
\begin{lemma}
\label{lem:SumIntegralLambdaTilde}
For $m \gtrsim 1$ we have
\begin{equation}
\lk |\int_{\R^3} \lambda_{\tilde s,Q_\mu,K,m,\delta}(\tilde t) \dI \tilde t - \LatFact^3 \sum_{\tilde t \in \Lat + AK} \lambda_{\tilde s,Q_\mu,K,m,\delta}(\tilde t) \rk | \lesssim  \frac 1m \left( \frac 1 {\ell Q_\mu} + \frac 1 {\delta^{1/2}}\right) \lk ( 1 +  \frac {N \delta}{ \lE^{2} Q_\mu^2 } + \frac {1}{\delta \ell Q_\mu } + \frac {N}{ \lE^3 Q_\mu^3} \rk)  \,.
\end{equation}
\end{lemma}

\begin{proof}
As in the proof of the previous Lemma, we denote $\lambda(\tilde t) = \lambda_{\tilde s,Q_\mu,K,m,\delta}(\tilde t)$, and write it as 
\begin{align}\nonumber
\lambda(\tilde t) & = c_5 ( (\tilde s - AK)^2 + 2 Q_\mu^2 + N \delta \lE^{-2}) (c_1 \tilde s^2 + c_2 Q_\mu^2 + c_3 K^2)^{-1/4} \\
&\quad \times \frac 1 {(\tilde t - AK)^2 + \delta \lE^{-2}} (c_1 \tilde t^2 + c_2 Q_\mu^2 + c_3 K^2)^{-1/4} \frac {|\tilde s \cdot \tilde t|}{(\tilde s^2 + \tilde t^2 + c_2 Q_\mu^2 + c_3 K^2)^2 - (c_4 \tilde s \cdot \tilde t)^2} \label{eq:Form1} 
\end{align}
with appropriate  coefficients $c_1,c_2,c_3,c_4,c_5$ depending on $m$. Its gradient equals  
\begin{align}\nonumber
\nabla \lambda(\tilde t) &= \ub{ -2 \frac {\tilde t  -  A  K}{(\tilde t-A K)^2 + \delta \lE^{-2}} \lambda(\tilde t) }{\RN 1} - \ub{ \frac 1 2 \frac {c_1 \tilde t }{c_1 \tilde t^2 + c_2 Q_\mu^2 + c_3 K^2} \lambda(\tilde t) } {\RN 2} \\ & \quad 
  - \ub{\frac {4 \tilde t  (\tilde s^2 + \tilde t^2 + c_2 Q_\mu^2 + c_3 K^2) - 2 c_4^2 \tilde s  (\tilde s \cdot \tilde t)}{(\tilde s^2 + \tilde t^2 + c_2 Q_\mu^2 + c_3 K^2)^2 - (c_4 \tilde s \cdot \tilde t)^2 } \lambda(\tilde t)}{\RN 3} + \ub{\frac {\tilde s } {\tilde t \cdot \tilde s} \lambda(\tilde t)}{\RN 4} \,.
\end{align}
We can quantify the difference between the Riemann sum and the integral by
\begin{equation}
\lk |\int_{\R^3} \lambda(\tilde t) \dI \tilde t - \LatFact^3 \sum_{\tilde t \in \Lat + AK} \lambda(\tilde t) \rk | \lesssim \lE^{-4}  \sum_{\tilde t \in \Lat + AK} \max_{\tau \in \CB(\tilde t)} |\nabla \lambda(\tau)|\,.
\label{eq:approxRiemann} 
\end{equation}
With the aid of the triangle inequality we can treat the terms $\RN 1 - \RN 4$ separately. 

We can bound $\RN 1$ as 
\begin{equation}
|\RN 1| \le \frac {2}{\sqrt{(\tilde t - A K)^2 + \delta \lE^{-2}}} \lambda(\tilde t) \le \frac {2 \lE}{\delta^{1/2}} \lambda(\tilde t)  \,.
\end{equation}
For the second term we obtain
\begin{equation}
|\RN 2| \le \frac 1 2 \sqrt{\frac{c_1}{c_2}} \frac {1}{Q_\mu} \lambda(\tilde t) = \frac 1 {2^{3/2}} \sqrt{ \frac {m+2}{m+1}} \frac {1}{Q_\mu} \lambda(\tilde t)  \lesssim \frac 1{Q_\mu} \lambda(\tilde t) \,.
\end{equation}
For  $\RN 3$, we use similar estimates as in Lemma \ref{lem:pointwiseBound} to get
\begin{equation}
|\RN 3| \lesssim  \frac {|\tilde t| + |\tilde s|} {\tilde s^2 + \tilde t^2 + c_2 Q_\mu^2 + c_3 K^2}  \lambda(\tilde t) \lesssim \frac 1 {Q_\mu} \lambda(\tilde t) \,.
\end{equation}
Finally, for $\RN 4$ we have to proceed slightly differently. If we use
\begin{equation}
|\tilde s|  \le \frac 1 {2 \sqrt c_2 Q_\mu}  (\tilde s^2 + \tilde t^2 + c_2 Q_\mu^2 + c_3 K^2) 
\end{equation}
instead of \eqref{4.59}, we see that we can bound $|\RN 3|$ from above by $Q_\mu^{-1}$ times the right side of \eqref{eq:pointWiseDet}. 
Using  Lemma~\ref{lem:pointwiseBound}  we conclude that
\begin{align}\nonumber
\eqref{eq:approxRiemann} &\le  \lE^{-4}  \sum_{\tilde t \in \Lat + AK} \max_{\tau \in \CB(\tilde t)} \lk(|\RN 1| + |\RN 2| + |\RN 3| + |\RN 4| \rk) \\ &\lesssim  \frac 1 m \lk ( \frac {1} {\lE Q_\mu} + \frac 1 {\delta^{1/2}} \rk) \lk ( 1 +  \frac {N \delta}{ \lE^{2} Q_\mu^2 } + \frac {1}{\delta \ell Q_\mu } + \frac {N}{ \lE^3 Q_\mu^3} \rk)\,.
\end{align}
Here we have used that the bound \eqref{eq:sum} holds also with $\lambda_{\tilde s, Q_\mu, K , m ,\delta}$ replaced by the right side of \eqref{eq:pointWiseDet}, as shown in the proof of Lemma~\ref{lem:pointwiseBound}. 
This completes the proof.
\end{proof}

\begin{lemma}\label{lem:cL} \label{cor:DiffLambda} 
There exists a $c_\Lambda > 0$ such that 
\begin{equation}
\tilde \Lambda(m,\kappa) \le \Lambda(m) + \frac 1m \frac{c_\Lambda}{(1-\kappa/c_T)^2}  \,  N^{-2/9}
\end{equation}
whenever $\kappa < \cT$ and $\Lambda(m)\leq 1$, where $\cT$ is defined in \eqref{def:cT}.
\end{lemma}

\begin{proof}
We first note that  $\Lambda(m)\leq 1$ implies $m\gtrsim 1$. Moreover, from the definition \eqref{def:lambda} we have
\begin{equation}
\ \lambda_{\tilde s,Q_\mu,K,m,\delta}(\tilde t) \le \lk (1 + \frac { N\delta }{ 2 \lE^2 Q_\mu^2 } \rk )  \lambda_{\tilde s,Q_\mu,K,m,0}(\tilde t) \,.
\end{equation}
Combining this with Lemma~\ref{lem:SumIntegralLambdaTilde} and taking the supremum over $\tilde s$, $K$ and $Q_\mu^2 \geq (\cT-\kappa) N^{5/3} \ell^{-2}$, we obtain
\begin{equation}
\tilde \Lambda(m,\kappa) -  \Lambda(m) \lesssim  \frac 1m \inf_{\delta>0} \sup_{ Q_\mu^2 \geq (\cT-\kappa) N^{5/3} \ell^{-2}} \left[ \frac {N \delta}{ \lE^{2} Q_\mu^2 } +  \left( \frac 1 {\ell Q_\mu} + \frac 1 {\delta^{1/2}}\right) \lk ( 1 +  \frac {N \delta}{ \lE^{2} Q_\mu^2 } + \frac {1}{\delta \ell Q_\mu } + \frac {N}{ \lE^3 Q_\mu^3} \rk)  \right]
\end{equation}
where we also used that $\Lambda(m)\lesssim m^{-1}$ for $m\gtrsim 1$. The supremum over $Q_\mu$ is clearly achieved for $Q_\mu^2 = (\cT-\kappa)N^{5/3}\ell^{-2}$. For an upper bound, we shall choose $\delta \sim  N^{4/9}$, which yields the desired bound
\begin{equation}
\tilde \Lambda(m,\kappa) -  \Lambda(m) \lesssim  \frac 1m (\cT-\kappa)^{-2} N^{-2/9}\,.
\end{equation} 
\end{proof}

\subsection{Proof of Theorem \ref{thm:AprioriBound}} 
Using Prop.~\ref{lem:ToffEst}, Eq.~\eqref{cons44} and Lemma~\ref{cor:DiffLambda}, we get the lower bound
\begin{align}\nonumber
 N^{-1} \TSP(\xi^{\rm per}) &\ge \left( \frac {2m\alpha }{m+1}  - \frac 1 {N^{5/3}\ell} \frac{c_L'}{\cT-\kappa} \right)  \norm{\xi^{\rm per}}^2 \\
&\quad+ \frac 1{1- \kappa/\cT}   \lk(1-  \kappa/\cT-\Lambda(m) - \frac{c_\Lambda N^{-2/9}}{ m(1-\kappa/c_T)^2} \rk)  \LatFact^{3N} \sum_{\vec q \in \Lat^N} \L(\vec q) |\xi^{\rm per}(\vec q)|^2   \label{eq:apri1} 
\end{align}
for any $0<\kappa< \cT$ and  $\mu \geq -\kappa N^{5/3}/\ell^2$. 
Note that the coefficient in front of the last sum is positive for all $N>N_0(\kappa,m)$, defined in \eqref{def:n0}. If $\alpha$ is large enough such that also the first term on the right side of \eqref{eq:apri1} is non-negative, we conclude that $ \TSP(\xi^{\rm per}) \geq 0$.

In case $2m \alpha < (m+1) c_L' (\cT-\kappa)^{-1}N^{-5/3} \ell^{-1}$, on the other hand, we need to dominate the first term on the right side of \eqref{eq:apri1} by the second.   We  use   \eqref{def:cT} to obtain the lower bound 
\begin{equation}
\L (\vec q) \ge 2 \pi^2 \lk ( \frac {2m} {m+1} \rk )^{3/2} Q_\mu \geq  2 \pi^2 \lk ( \frac {2m} {m+1} \rk )^{3/2}  \sqrt{\mu + \kappa N^{5/3} \lE^{-2}}\, . 
\end{equation}
In particular, if we choose
\begin{equation}
\mu =
-\kappa N^{5/3} \lE^{-2} + \frac 1 {4 \pi^4} \frac {m+1} {2 m}   \dfrac {(1-\kappa/\cT)^2[\alpha-(2m)^{-1}(m+1) c_L' (\cT-\kappa)^{-1}N^{-5/3} \ell^{-1}]^2_-}{(1-\kappa/\cT- \Lambda(m)- c_\Lambda m^{-1} (1-\kappa/\cT)^{-2} N^{-2/9})^2} \label{eq:muapriori} 
\end{equation} 
we again conclude that $\TSP(\xi^{\rm per}) \geq 0$. 

Note that for our choice of $\mu$, satisfying in particular $\mu \geq - \cT N^{5/3}\ell^{-2}$, we have
\begin{equation}
\int_{\BE^{N+1}}  \lk (  |\tnabla \phi_\mu^\per|^2    + \mu  |\phi_\mu^\per|^2 \right ) \geq 0
\end{equation}
for all $\phi_\mu^\per \in H^1_\per (B^{N+1})$ that are antisymmetric in the last $N$ variables. Hence the positivity of $\TSP(\xi^\per)$ implies that $\QFP(\psi^\per) \geq - \mu \|\psi^\per\|^2$.  In combination with Lemmas~\ref{lem:periodic} and~\ref{cor:nmu}, 
this completes the proof of Theorem~\ref{thm:AprioriBound}. To simplify its statement, we have additionally used that
\begin{equation}
 (1-\kappa/\cT)^2[\alpha-(2m)^{-1}(m+1) c_L' (\cT-\kappa)^{-1}N^{-5/3} \ell^{-1}]^2_- \leq [\alpha-(2m)^{-1}(m+1) c_L' \cT^{-1}\ell^{-1}]^2_-
 \end{equation}
for $N\geq 1$, and defined
\begin{equation}\label{fdef:cL}
c_L \coloneqq \frac{\mcritMS+1}{2\mcritMS} \frac {c_L'}\cT 
\end{equation}
in Eq.~\eqref{lbthm41}, where $\mcritMS \approx 0.36$ is chosen such that $m \geq \mcritMS$ for $\Lambda(m)\leq 1$. 
\hfill\qed

\section{Proof of Theorem \ref{thm:main}}
\label{sec:Loc}
\newcommand{\boundary}{ \pd}
\newcommand{\inbound}{\supp }
In this section we will give the proof of our main result, Theorem~\ref{thm:main}.

Let $\B = (0,\lO)^3$ and $\bar \B = \bigcup_i^M \bar B_i$ a disjoint decomposition into cubes $B_i = (0,\lI)^3 + z_i$ with $z_i \in \R^3$. We will choose $\lI$ such that $\lO/\lI \in \N$ in which case $M = (\lO/\lI)^3$. Let $1/4 > \ve > 0$ and let $\eta \in \CC^\infty_0(B_\ve(0))$ be non-negative, where we denote by $B_\ve(0)$ the centered ball of radius $\ve$. In the following we will assume that $\ve$ is  a fixed constant independent of all parameters (for example $\ve = 1/8$  works). For $x \in \B$,  define
\begin{equation}
J_i(x) = \lk( \frac {\int_{B_i} \eta(\lI^{-1}(x-y)) \dI y}{\int_{\B} \eta(\lI^{-1}(x-y)) \dI y} \rk)^{1/2}  \,.
\end{equation}
Then $\supp J_i \subseteq B_i + B_{\lI \ve}(0)$ and $J_i(x) = 1$ for $x \in \lI(\ve, 1-\ve)^3 + z_i$. Moreover, $\sum_{i=1}^M J_i^2(x) = 1$ for $x \in B$ by construction. The derivative of $J_i$ can be bounded uniformly in $i$ and $M$ by a constant $c_\eta$  depending only on $\eta$ (and hence $\ve$) as 
\begin{equation}
|\nabla J_i|^2 \le \frac {c_\eta}{\lI^2}\,. \label{def:ceta} 
\end{equation}

Let $\psi \in D(\QF)$ be such that $\supp \psi \subseteq \B^{N+1}$ and $\norm{\psi}_2 = 1$. 
 We use the IMS formula, Prop.~\ref{thm:IMS},  for the quadratic form $\QF$ to localize the impurity particle (with coordinate $x_0$). With $J_i \psi$ denoting the function $(J_i \psi)(x_0,\vec x) = J_i(x_0) \psi(x_0,\vec x)$ we  obtain
\begin{equation}
\QF(\psi) = \sum_{i=1}^M  \QF(J_i \psi) - \frac 1 {2m} \sum_{i=1}^M \int |\nabla J_i(x_0)|^2 |\psi(x_0,\vec x)|^2 \dI x_0 \dI \vec x \,. \label{eq:QF1} 
\end{equation}
We note that the last term is bounded by
\begin{equation}\label{eq:QF1a}
 \sum_{i=1}^M \int |\nabla J_i(x_0)|^2 |\psi(x_0,\vec x)|^2 \le \frac { c_\eta} {\lI^2} \sum_{i=1}^M \int_{\boundary J_i} |\psi(x_0,\vec x)|^2 \dI x_0 \dI \vec x \le \frac {8 c_\eta} {\lI^2}
\end{equation}
since $\ve < 1/2$, where $\boundary J_i = \supp |\nabla J_i|$.  Recall  the definition of the mean density, $\dens = N \lO^{-3}$. We will choose $\lI \sim \dens^{-1/3}$ which means that \eqref{eq:QF1} is of the order $\dens^{2/3}$.

In the next step we want to localize the other particles, to be able to distinguish whether they are close to the impurity or far from it. Because we violate the antisymmetry constraint by doing so, we will work with the extended quadratic form $\QFE$ defined in \eqref{eq:quadraticFormNoSym}. Let $V\in\CC^\infty_0(\R^3)$ satisfy $0\leq V\leq 1$, with $\supp V \subseteq [-2 \ve, 1 + 2\ve]^3$ and $V(x) = 1$ for $ x \in [-\ve,1+\ve]^3$. We define $V_i(x) = V((x-z_i)/\lI)$ and  $\tilde V_i(x) \coloneqq \sqrt{ 1- V_i(x)^2}$. Figure \ref{fig:sketch} visualizes this setup.

\begin{figure}[t]
\begin{center}
\includegraphics[width=13cm]{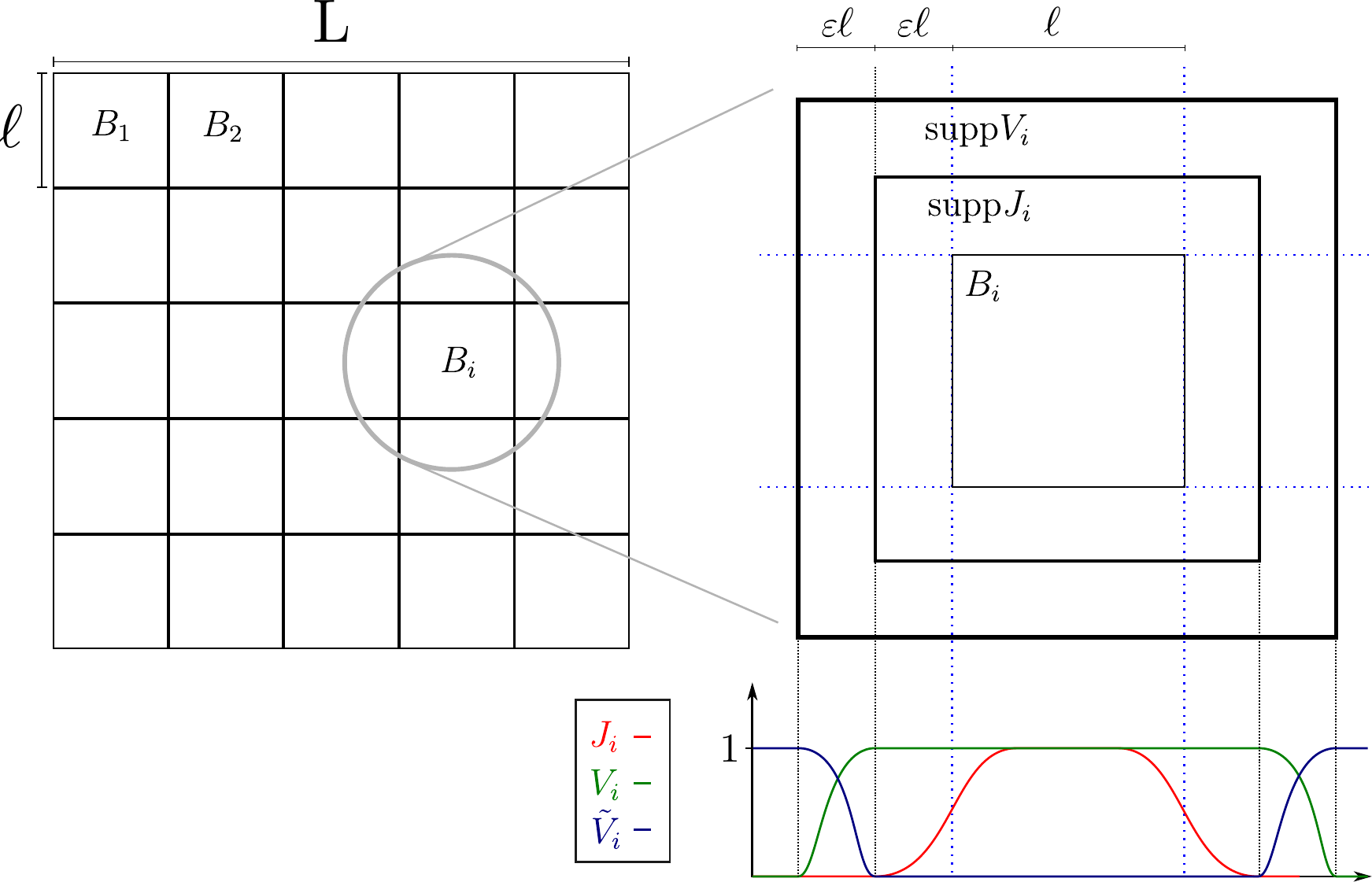}
\end{center}
\caption{A sketch of the setup, the partitions $J_i, V_i, \tilde V_i$ and their  boxes of support.}
\label{fig:sketch}
\end{figure} 

We localize all the remaining particles using the IMS formula  in Prop.~\ref{thm:IMS}, with the localization functions
\begin{equation}
(x_1,\ldots, x_N) \mapsto \prod_{j \in A} V_i(x_j) \prod_{k \in A^c} \tilde V_i(x_k) 
\end{equation}
for $A \subseteq \{1,\ldots, N\}$, where $A^c = \{1,\ldots, N\} \setminus A$. For short we define
\begin{equation}
\vp_{i,A}(x_0,\vec x) \coloneqq  J_i(x_0) \prod_{k \in A} V_i(x_k) \prod_{j \in A^c} \tilde V_i(x_j) \psi(x_0,\vec x)\,. 
\end{equation}
A straightforward calculation using Prop.~\ref{thm:IMS} and the fact that $V_i^2 + \tilde V_i^2 = 1$ shows that
\begin{equation}
\QF(J_i\psi) = \sum_{A \subseteq \{1,\ldots, N\} }  \left( \QFE \lk(\vp_{i,A} \rk) 
 - \frac 1 2  \sum_{j=1}^N   \int \left(|\nabla V_i(x_j)|^2 + |\nabla \tilde V_i(x_j)|^2\right) \lk | \vp_{i,A}(x_0,\vec x) \rk |^2  \dI x_0 \dI \vec x \right) \, .  \label{eq:afterIMS} 
\end{equation}
Here it is necessary to 
 introduce  the extended quadratic form $\QFE$ since  the functions $\vp_{i,A}$ are not antisymmetric in all $N$ variables $(x_1,\dots,x_N)$. 
They are still separately antisymmetric in the coordinates  in $A$ and in the ones in $A^c$, however.

In the next lemma we will show that the energy  $\QFE(\varphi_{i,A})$ splits up into a non-interacting energy  for the particles in $A^c$ that are localized away from the impurity, and in a point interacting quadratic form for particles in $A$.

\begin{lemma}
\label{lem:OutsideFreePart}
We define the functions $\vpIn \in L^2(\R^{3(|A|+1)})$ and $\vpOut \in L^2(\R^{3|A^c|})$ via their Fourier transforms as
\begin{equation}
\hvpIn(p_0,\vec p_{A}) = \hat\vp_{i,A}(p_0, \vec p) = 
\hvpOut(\vec p_{A^c})\,.
\end{equation}
Then 
\begin{equation}
\QFE(\vp_{i,A}) = \int \QFp{\alpha}{|A|}(\vpIn) \dI \vec p_{A^c} + \int \left\langle \vpOut \left| - \tfrac 12 \sum\nolimits_{i\in A^c} \Lap_i \right| \vpOut\right\rangle   \dI \vec p_{A} \dI p_0  \,.\label{eq:splitUpInsideOutside} 
\end{equation}
\end{lemma}

\begin{proof}
We define $\xi_j$ and $\phi_\mu$ for some $\mu > 0$ using the unique decomposition $\vp_{i,A} = \phi_\mu + \sum_{j=1}^N G_\mu \xi_j$.  Corollary \ref{cor:domainXi} implies that $\xi_j = 0$ for $j\in A^c$. 
 Hence 
\begin{align}\nonumber 
\QFE(\vp_{i,A}) &= \int \dI \vec p_{A^c} \Bigg [ \int |\hat \phi_\mu(p_0,\vec p)|^2 \lk (\frac 1 {2m} p_0^2 + \frac 1 2  \vec p^2 + \mu \rk) \dI \vec p_A \dI p_0  - \mu \int |\hat \vp_{i,A}(p_0,\vec p)|^2 \dI \vec p_A \dI p_0 \\ \nonumber
&\qquad\qquad\quad + \frac {2m}{m+1} \alpha \sum_{i\in A} \int |\hat\xi_i(p_i, \hat p_i)|^2  \dI \vec p_A + \sum_{i\in A} \int \L(p_i, \hat p_i) |\hat\xi_i(p_i, \hat p_i)|^2 \dI \vec p_{A}  \\
& \qquad\qquad\quad - \sum_{ \substack{i,j \in A \\ i \ne j }} \int \frac { \hat\xi_i^\ast(p_0+p_i, \hat p_i) \hat\xi_j(p_0+p_j, \hat p_j)} {\frac 1 {2m} p_0^2 + \frac 1 2 \vec p^2 + \mu} \dI \vec p_A \dI p_0 \Bigg]   \,.
\end{align}
Following the argumentation in the proof of Lemma~\ref{cor:nmu} we see that the expression inside the integral over $\vec p_{A^c}$ is independent of $\mu$. In particular this allows us to shift $\mu \to \mu - \vec p_{A^c}^2/2$ for fixed $\vec p_{A^c}$, which gives
\begin{align}\nonumber
\QFE(\vp_{i,A}) &= \int \dI \vec p_{A^c} \Bigg [ \int |\hat \phi_{\mu-\vec p_{A_c}^2/2}(p_0,\vec p)|^2 \lk (\frac 1 {2m} p_0^2 + \frac 1 2  \vec p_A^2 + \mu \rk) \dI \vec p_A \dI p_0  \\ \nonumber
& \qquad\qquad\quad - \lk (\mu - \frac {\vec p_{A^c}^2}2 \rk) \int |\hat \vp_{i,A}(p_0,\vec p)|^2 \dI \vec p_A \dI p_0 \\ \nonumber
&\qquad\qquad\quad + \frac {2m}{m+1} \alpha \sum_{i\in A} \int |\hat\xi_i(p_i, \hat p_i)|^2  \dI \vec p_A + \sum_{i\in A} \int \Lp {\mu}{|A|}(p_i, \vec p_{A \setminus \{i\}}) |\hat\xi_i(p_i, \hat p_i)|^2 \dI \vec p_{A} \\
&\qquad\qquad\quad - \sum_{ \substack{i,j \in A \\ i \ne j }} \int \frac { \hat\xi_i^\ast(p_0+p_i, \hat p_i) \hat\xi_j(p_0+p_j, \hat p_j)} {\frac 1 {2m} p_0^2 + \frac 1 2 \vec p^2_A + \mu} \dI \vec p_A \dI p_0 \Bigg]   
\end{align}
where we used the fact that  $\Lp{\mu - \vec p_{A^c}^2/2} {N} (p_i, \hat p_i) = \Lp{\mu}{|A|}(p_i, \vec p_{A \setminus \{i\}})$. The result then follows by noting that the Fourier transform of the regular part of $\vpIn$ for fixed $\vec p_{A^c}$ is equal to $\hat\phi_{\mu-\vec p_{A^c}^2}(\, \cdot\, ,\vec p_{A^c})$, and using the 
the antisymmetry of $\vpIn$. 
\end{proof}

We can apply a similar decomposition also to the second term in \eqref{eq:afterIMS}. For simplicity, let
\begin{equation}
W_i(x) = \frac 12 \left( |\nabla V_i(x)|^2 + |\nabla \tilde V_i(x)|^2 \right) \,.
\end{equation}
Then \eqref{eq:afterIMS} and \eqref{eq:splitUpInsideOutside} imply that we can write
\begin{equation}
\QF(J_i\psi)  = \sum_{A \subseteq \{1,\ldots, N\} }  \| \vp_{i,A}\|^2  \left[  \mathfrak{A}_{i,A}   +  \mathfrak{B}_{i,A}
  \right]  \label{eq:afterIMS2} 
\end{equation}
where
\begin{equation}
\mathfrak{A}_{i,A}=  \| \vp_{i,A}\|^{-2}  \int \left(  \QFp{\alpha}{|A|}(\vpIn) -  \left\langle \vpIn \left|  \sum\nolimits_{j \in A}  W_i(x_j)  \right| \vpIn\right\rangle \right)   \dI \vec p_{A^c}
\end{equation}
and
\begin{equation}
\mathfrak{B}_{i,A}=  \| \vp_{i,A}\|^{-2}  \int \left\langle \vpOut \left|  \sum\nolimits_{j\in A^c} \left( - \tfrac 12 \Lap_j + W_i(x_j) \right) \right| \vpOut\right\rangle   \dI \vec p_{A} \dI p_0 \,.
\end{equation}

To obtain a lower bound on $\mathfrak{A}_{i,A}$ we can use Theorem \ref{thm:AprioriBound}, and for the non-interacting part $\mathfrak{B}_{i,A}$  we use the following proposition. We recall that the  energy $\EGDp{n}$ on the box $\B = (0, \lO)^3$ was defined in the beginning of Section \ref{sec:Results} as the ground state energy of the non-interacting Hamiltonian $\HGp{n}$ with Dirichlet boundary conditions.

\begin{prop}
\label{lem:LowerBoundOutside}
For $n\in \N$, let $\phi \in H_{\rm as}^1(\R^{3n})$ be supported in $(0,L)^{3n}$, with $\norm{\phi}_2 = 1$, and let $1\leq i\leq M$. Then 
\begin{equation}
\sum_{j=1}^n   \int  \left( \tfrac 12  |\nabla_j \phi|^2  -   W_i(x_j) |\phi|^2 \right)
 \ge \EGDp{n} - \const \, \lk ( \frac {n^{1/3}}{\ell L}  + \lI^{-2} + \frac {n \lI}{L^3} \rk )\, . \label{eq:LowerBoundOutside} 
\end{equation}
\end{prop}

\begin{proof}
The result follows in a straightforward way from Corollary \ref{cor:BoundPotential}, which is an adaptation of the Lieb-Thirring inequality at positive density derived in  \cite{posDensitiy}.  
We  use that $|\supp(W_i)| \lesssim \lI^3$ and $\|W_i\|_\infty \lesssim \lI^{-2}$. 
This allows us to bound the right side of \eqref{eq:cora1} as
\begin{equation}
\int_B \lk ( \frac {n^{1/3}}L  |W|^2 + |W|^{5/2} + \frac n{L^3} |W| \rk ) \lesssim  \frac {n^{1/3}}{\ell L}  + \lI^{-2} + \frac {n \lI}{L^3}
 \label{eq:errTermCor} 
\end{equation}
from which the statement readily follows.
\end{proof}

Since $\vpOut$ is an antisymmetric function supported in $B^{|A^c|}$, Prop.~\ref{lem:LowerBoundOutside} implies that 
\begin{equation}
\left\langle \vpOut \left|  \sum\nolimits_{j\in A^c} \left( - \tfrac 12 \Lap_j + W_i(x_j) \right) \right| \vpOut\right\rangle \geq \left( E^D_{|A^c|} - \const  \, \lk (  {\bar\rho^{1/3}}{\ell^{-1}}  + \lI^{-2} +  \bar\rho \lI \rk ) \right) \| \vpOut\|^2 
\end{equation}
where we used $|A^c|\leq N$ in the error term. 
To minimize the error we choose  $\lI \sim \dens^{-1/3}$. The factor on the right side of \eqref{eq:LowerBoundOutside} then equals $\EGDp{N-|A|} -  \const \, \dens^{2/3}$.
Because of the condition that $\lO/\lI \in \mathbb N$ we cannot choose $\lI$ without restriction but it is always possible to choose a value such that $\lI \sim \dens^{-1/3}$. 
We define $e_N$ to be the $N$-th eigenvalue of the one-particle Dirichlet Laplacian on $B=(0,L)^3$. Then $\EGDp {N-|A|} \ge \EGDp N  - |A| e_N$. Moreover, we can bound $e_N \lesssim \dens^{2/3}$. In particular, 
\begin{equation}\label{519}
\mathfrak{B}_{i,A} \geq  \EGDp{N} -  \const \left( |A|+1 \right)  \dens^{2/3}   \,.
\end{equation}

We proceed with a lower bound on $\mathfrak{A}_{i,A}$.  
Theorem \ref{thm:AprioriBound} can be used for a lower bound on $ \QFp{\alpha}{|A|}$  only if $|A| > N_0$, with $N_0$ defined in \eqref{def:n0}.
In  case that $|A| \leq 2 N_0$ we use the bound \eqref{eq:bound3}  originating form \cite{MoserSeiringer2017} instead, which implies that
\begin{equation}
\QFp{\alpha}{|A|}(\vpIn)   \gtrsim -\frac {\alpha_-^2}{(1- \Lambda(m))^2} \norm{\vpIn}^2 \label{eq:LowerBoundLowN} 
\end{equation}
using $m \gtrsim 1$.
In combination with $\|W_i\|_\infty \lesssim \bar\rho^{2/3}$ this gives the lower bound
\begin{equation}
\mathfrak{A}_{i,A} \gtrsim -\frac {\alpha_-^2}{(1- \Lambda(m))^2}  - |A| \bar\rho^{2/3} 
\end{equation}
and hence
\begin{equation}\label{5:22}
\mathfrak{A}_{i,A} + \mathfrak{B}_{i,A} \geq  \EGDp{N} - \const \left( \frac {\alpha_-^2}{(1- \Lambda(m))^2}  + (N_0+1) \bar\rho^{2/3} \right)
\end{equation}
in case $|A|\leq 2 N_0$. 

For $|A|\geq 2 N_0$, we use the bound in Theorem \ref{thm:AprioriBound} on $ \QFp{\alpha}{|A|}(\vpIn)$. Since $\vpIn$ is an $|A|+1$-particle wavefunction supported in a cube of side length $\ell(1+2\ve)$, Theorem \ref{thm:AprioriBound} implies that
\begin{equation}
 \QFp{\alpha}{|A|}(\vpIn) \geq  \left(  \kappa \frac {|A|^{5/3}}{\lI^2 (1 + 2 \ve)^2} - U  \right) \| \vpIn \|^2 
\end{equation}
with 
\begin{equation}
U =  \frac 1 {4 \pi^4} \frac {m+1} {2 m}   \dfrac { [\alpha-c_L \lE^{-1}]^2_-}{(1-\kappa/\cT - \Lambda(m))^2 ( 1 - 2^{-2/9} )^2 } \,.
\end{equation}
In combination with \eqref{519} and $\|W_i\|_\infty \lesssim \bar\rho^{2/3}$ this yields the bound
\begin{align}\nonumber
\mathfrak{A}_{i,A} + \mathfrak{B}_{i,A} & \geq  \EGDp{N}  +  \kappa \frac {|A|^{5/3}}{\lI^2 (1 + 2 \ve)^2} -  \const \left( |A|+1 \right)  \dens^{2/3}  - U
\\ & \geq    \EGDp{N}  - U - \const \, \kappa^{-3/2} \bar\rho^{2/3} \label{5:25}
\end{align}
where we have minimized over $|A|$ in the last step, and used that $\ve\lesssim 1$ and $\ell \sim \bar\rho^{-1/3}$.

We are still free to choose $\kappa$ in such a way as  to minimize the error terms. We shall choose $\kappa = \cT \nu (1-\Lambda(m))$ for some $0<\nu<1$ (e.g., $\nu = 1/2$). Then $N_0\lesssim (1-\Lambda(m))^{-9/2}$, and hence \eqref{5:22} and \eqref{5:25} together yield the bound
\begin{align}\nonumber
\mathfrak{A}_{i,A} + \mathfrak{B}_{i,A}  & \geq  \EGDp{N}  - \const \left( \frac {[\alpha - c_L \ell^{-1} ]_-^2}{(1- \Lambda(m))^2} + \frac {\dens^{2/3}} {(1- \Lambda(m))^{9/2}} \right)
\\  & \geq  \EGDp{N}  - \const \left( \frac {\alpha_-^2}{(1- \Lambda(m))^2} + \frac {\dens^{2/3}} {(1- \Lambda(m))^{9/2}} \right)
\end{align}
which is valid for all $A \subset \{ 1, \dots, N\}$. In combination with \eqref{eq:QF1}, \eqref{eq:QF1a} and \eqref{eq:afterIMS2}, 
this completes the proof of  Theorem~\ref{thm:main}. 
\hfill\qed

\appendix

\section{Lieb-Thirring inequality in a box}
In this appendix we will follow the analysis of \cite{posDensitiy} to show a positive density Lieb-Thirring inequality for a system of non-interacting fermions in a box with Dirichlet boundary conditions. When reformulated via a Legendre transformation as a bound on the difference between the ground state energies with and without an external potential, we will see that this inequality in particular implies Prop. \ref{lem:LowerBoundOutside}. 

Let $\BLT = [-\lLT/2,\lLT/2]^3$ be the cube in $\R^3$ and let $\Pi^-_{\lLT,\mu} \coloneqq  \unit(- \Delta_\lLT \le \mu)$, where $\Delta_\lLT$ denotes the Dirichlet Laplacian on  $\BLT$.  For short we will just write $\Pi^-$ for $\Pi^-_{\lLT,\mu}$,  and $\Pi^+ = 1- \Pi^-$. For a density matrix $\gamma$ we denote the corresponding density by $\rho_\gamma$. Of particular relevance for us is the density corresponding to $\Pi^-$, which we denote by $\rho_0$. 
Differently to the case of periodic boundary conditions (discussed in \cite{posDensitiy}), $\rho_0$ is not a constant and is given by
\begin{equation}
\rho_0(x) = \sum_{\substack{p \in \pi \mathbb N^3/\lLT\\ p^2 \le \mu}} |\phi_p(x)|^2 \label{def:rho0} 
\end{equation}
where $\phi_p$ are the eigenvectors of $-\Lap_\lLT$ to the eigenvalues $p^2$, i.e.,
\begin{equation}
\phi_p(x) =  \lk ( \frac {2} {\lLT} \rk )^{3/2} \prod_{j=1}^3 \cos(p_j x_j) \label{eq:DirichletState} 
\end{equation}
for $x=(x_1,x_2,x_3) \in \R^3$.
Since the absolute value of each eigenvector is pointwise bounded by $(2/\lLT)^{3/2}$ we have
\begin{equation}
\rho_0(x) \le  \lk ( \frac 2 \lLT \rk )^{3}  \sum_{\substack{p \in \pi \mathbb N^3/\lLT\\ p^2 \le \mu}} 1 \le \lk ( \frac 2 \lLT \rk )^{3} \frac {4 \pi}{3} \frac {\mu^{3/2} \lLT^3}{\pi^3} = \frac {2^5 \mu^{3/2}}{3 \pi^2}  \,.
\end{equation}

\begin{remark*}
Since the lowest eigenvalue of $-\Delta_L$ equals $3 \pi^2 \lLT^{-2}$, the problem simplifies for 
 $\mu < 3 \pi^2 \lLT^{-2}$ since the projections $\Pi^{\pm}_{\lLT,\mu}$ become trivial. In this case we can simply apply the original Lieb-Thirring inequality \cite{LT} to obtain the desired bound. For our application we shall need $\mu \gg L^{-2}$, however, hence we shall restrict our attention to $\mu \ge 3 \pi^2 \lLT^{-2}$ in the following theorem.
\end{remark*}

For a real number $t$ we denote its positive part by $t_+$ and its negative part by $t_-$. In particular, $t = t_+ - t_-$. 

\begin{theorem}
\label{thm:posDensity}
Let $\mu \ge 3 \pi^2 \lLT^{-2}$. Let $Q$ be a self-adjoint operator of finite rank satisfying $- \Pi^-_{\lLT,\mu} \le Q \le 1 - \Pi^-_{\lLT,\mu}$,   with density $\rho_Q$. There exist positive constants $\tilde K$ and $\eta$ independent of $\mu, \lLT$ and $Q$ such that
\begin{equation}
\tr (- \Delta_\lLT - \mu) Q \ge \tilde K \int_{\BLT} S \lk ( (|\rho_Q(x)| - \eta \lLT^{-1} \mu)_+ \rk) \dI x 
\end{equation}
with
\begin{equation}
S(\rho) \coloneqq  (\mu^{3/2} + \rho)^{5/3} - \mu^{5/2} - \frac 5 3 \mu \rho\,. 
\end{equation}
\end{theorem}

\begin{remark*}
In \cite{posDensitiy} a similar result was proven for the Laplacian with periodic boundary conditions and we mostly follow that proof.
\end{remark*}

\begin{remark*}
The crucial properties of the function $S$ are its positivity and the fact that $S(\rho)$ behaves like $\mu^{-1/2} \rho^2$ for small $\rho$ and like $\rho^{5/3}$ for large $\rho$. For technical reasons  it will also be convenient that $S$ is convex.
\end{remark*}

Essential for the proof will be to separate a given $Q$ into $Q = (\Pi^++\Pi^-)Q(\Pi^++\Pi^-) \eqqcolon Q^{++}+Q^{+-}+Q^{-+}+Q^{--}$. The densities associated to $Q^{\pm\pm}$ will be denoted by $\rho^{\pm\pm}$.
Before we proceed with the proof of the theorem we show the following  Lemma.

\begin{lemma}
\label{lem:Q2Q}
Assume $\Pi^- \le Q \le 1- \Pi^-$. Then
\begin{equation}
\tr \lk ( |- \Lap_\lLT - \mu| Q^2 \rk ) \le \tr(- \Lap_\lLT - \mu) Q \,.
\end{equation}
\end{lemma}

\begin{proof}
We claim that $Q^2 \le Q^{++} - Q^{--}$, which follows from the condition on $Q$. In fact,
\begin{equation}
-\Pi^- \le Q \le 1- \Pi^- \folgt 0 \le Q+ \Pi^- \le 1 \folgt (Q+\Pi^-)^2 \le Q + \Pi^- \,. \label{eq:Q2Split} 
\end{equation}
Expanding the last inequality proves the claim. Hence 
\begin{align}\nonumber
\tr(|\Lap_\lLT + \mu| Q^2) & \le \tr(|\Lap_\lLT + \mu| Q^{++}) - \tr(|\Lap_\lLT + \mu| Q^{--}) \\ & = \tr((-\Lap_\lLT-\mu)Q^{++} ) + \tr((-\Lap_\lLT-\mu) Q^{--}) = \tr((-\Lap_\lLT-\mu)Q) \,.
\end{align}
\end{proof}

\begin{proof}[Proof of Theorem \ref{thm:posDensity}]
We shall treat $Q^{\pm\pm}$ separately and combine the various terms at the end using the convexity of $S$.
\begin{proofpart}$Q^{++}$, $Q^{--}$\end{proofpart}
We shall follow the method introduced by Rumin in \cite{Rumin2011}. 
With the aid of the spectral projections $P_e \coloneqq \unit(|\Lap_\lLT +\mu| \ge e)$ we have the layer cake representation
\begin{equation}
|\Lap_\lLT + \mu| = \int_0^\infty P_e \dI e \,.
\end{equation}
Let us assume that $\gamma$ is a smooth enough finite rank operator with $0 \le \gamma \le 1$. Then
\begin{equation}
\tr |\Lap_\lLT + \mu| \gamma = \int_0^\infty \dI e \, \tr(P_e \gamma P_e) = \int_0^\infty \dI e \int_{\BLT} \rho_e(x) \dI x 
\end{equation}
where $\rho_e$ denotes the density of the finite rank operator $P_e \gamma P_e$. For  a bounded measurable set $A$  we estimate
\begin{align}\nonumber
\int_A \rho_e(x) \dI x & = \tr( \unit_A P_e \gamma P_e) = \norm{\unit_A P_e \gamma^{1/2}}_{\mathfrak S_2}^2 \\ \nonumber &
\ge  \lk (\norm{\unit_A \gamma^{1/2}}_{\mathfrak S_2} - \norm{\unit_A P_e^\perp \gamma^{1/2}}_{\mathfrak S_2} \rk )^2_+ \\ &
= \lk ( \lk( \int_A \rho_\gamma \rk)^{1/2} - \norm{ \unit_A P_e^\perp \gamma^{1/2}}_{\mathfrak S_2}   \rk)^2_+ 
\end{align}
where $\rho_\gamma$ denotes the density of  $\gamma$ and we used the triangle inequality for the  Hilbert-Schmidt norm $\norm{\,\cdot\,}_{\mathfrak S_2}$. 
Because $\norm{\gamma} \le 1$ we further get
\begin{equation}
\norm{\unit_A P_e^\perp \gamma^{1/2}}^2_{\mathfrak S_2} = \tr(\unit_A P_e^\perp \gamma P_e^\perp \unit_A) \le \norm{\unit_A P_e^\perp}^2_{\mathfrak S_2} \norm{\gamma} \le |A| f(e) 
\end{equation}
with
\begin{align}\nonumber
f(e) &\coloneqq  \lk( \frac 2 \lLT \rk)^3 \sum_{\substack{p \in \pi \mathbb N^3/\lLT\\ |p^2 - \mu| < e}} 1 = \lk( \frac 2 \lLT \rk)^3 \sum_{\substack{n \in \mathbb N^3/2 \\ |\frac {4 \pi^2}{\lLT^2} n^2 - \mu| < e}} 1  \\
&= \lk( \frac 2 \lLT \rk)^3  \lk [ \lk | \mathbb N^3/2 \cap B  \lk (\frac {\lLT}{2 \pi} (\mu + e)^{1/2} \rk ) \rk | - \lk |\mathbb N^3/2 \cap \bar B\lk( \frac {\lLT}{2 \pi} (\mu - e)^{1/2}_+ \rk) \rk| \rk ]  
\end{align}
where $B(R)$ denotes the centered open ball with radius $R$ and $\bar B(R)$ its closure. Here we used 
\begin{equation}
\norm{\unit_A P_e^\perp}_{\mathfrak S_2}^2 = \sum_{\substack{p \in \pi \mathbb N^3/\lLT\\ |p^2 - \mu| < e}} \int_A |\phi_p(x)|^2 \dI x \le |A|  \sum_{\substack{p \in \pi \mathbb N^3/\lLT\\ |p^2 - \mu| < e}} \sup_{x \in A} |\phi_p(x)|^2  \le |A| f(e) 
\end{equation}
where we bounded the eigenfunction $\phi_p$ of $-\Delta_L$  to the eigenvalue $p^2$ by $|\phi_p(x)| \le (2/\lLT)^{3/2}$. Taking $A = B(R)+ x$ with $R \to 0$ we obtain the pointwise bound
\begin{equation}
\rho_e(x) \ge (\sqrt{\rho_\gamma(x)} - \sqrt{f(e)})_+^2 \,.
\end{equation}
Hence we get
\begin{equation}\label{a16}
\tr |\Lap_\lLT + \mu| \gamma \ge \int_{\BLT} \dI x \int_0^\infty \dI e ( \sqrt{\rho_\gamma(x)} - \sqrt{f(e)})_+^2 = \int_{\BLT} R(\rho_\gamma(x)) \dI x 
\end{equation}
with
\begin{equation}
R(\rho) \coloneqq \int_0^\infty  \lk(\sqrt{\rho} - \sqrt{f(e)}\rk)^2_+ \dI e \,.
\end{equation}
To obtain the desired result we have to analyze $R(\rho)$ in more detail. In the following we will use $C$ to denote a generic constant, whose value can change throughout  the computation. Obviously 
\begin{equation}
 \left|  \lk | \mathbb N^3/2 \cap B(R) \rk|  - \frac {4 \pi} 3 R^{3} \right| \lesssim \max(1, R^2) \label{eq:PointsSphere} 
\end{equation}
and the same statement holds if one takes the closure $\bar B(R)$ instead of $B(R)$.
For $0 < x < 1$ and $M>0$,
\eqref{eq:PointsSphere} allows us to bound
\begin{align}\nonumber 
& |\mathbb N^3/2 \cap B(M(1+x)^{1/2})| - |\mathbb N^3/2 \cap \bar B(M(1-x)^{1/2})| \\ \nonumber
& \le \frac {4 \pi M^3}{3} \lk( (1+x)^{3/2} - (1-x)_+^{3/2} \rk) + C \max(1,M^{2}) \\
& \lesssim M^3 x + \max(1,M^{2}) \,,  \label{eq:diffCounting} 
\end{align}
where we used $(1+x)^{3/2} - (1-x)_+^{3/2} \lesssim x$.
Applying \eqref{eq:diffCounting} to $f(e)$ for $e/\mu < 1$ we get
\begin{equation}
f(e) 
\lesssim  \mu^{1/2} e + \frac \mu \lLT 
\end{equation}
using that $\mu \gtrsim \lLT^{-2}$ by assumption. For $e \ge \mu$ we get
\begin{equation}
f(e) = \lk (\frac 2  \lLT \rk )^3 \lk | \mathbb N^3/2 \cap B  \lk (\frac {\lLT}{2 \pi} (\mu + e)^{1/2} \rk ) \rk | \le \frac C {\lLT^3} \lk ( \lLT^3 (\mu +e)^{3/2} \rk) \le C  e^{3/2}  \,.
\end{equation}
Combining both statements we have thus shown that 
\begin{equation}
f(e) \le C \lk ( \frac {\mu}{\lLT} + \mu^{1/2} e \unit(e \le \mu) + e^{3/2} \unit(e > \mu) \rk ) = u + g(e) 
\end{equation}
with
\begin{equation}
g(e) \coloneqq C e \max(\mu^{1/2},e^{1/2}) \ , \quad 
u \coloneqq C \frac {\mu}{\lLT} \,.
\end{equation}
Using the explicit form of $g$, one readily checks that 
\begin{equation}
R(\rho) = \int_0^\infty \lk ( \sqrt{\rho} - \sqrt{ f(e)} \rk )_+^2 \dI e \ge \int_0^\infty \lk ( \sqrt{\rho}  - \sqrt{u} - \sqrt{ g(e)} \rk )_+^2 \dI e \gtrsim S((\rho - 2 u)_+) \,, 
\end{equation}
where we have also used that $(\sqrt \rho - \sqrt u)^2_+ \ge \frac 12 (\rho - 2 u)_+$. 
In combination with \eqref{a16}, this shows that 
\begin{equation}
\tr|-\Lap_\lLT-\mu|\gamma \gtrsim \int_{\BLT} S((\rho_\gamma(x) - C \lLT^{-1} \mu)_+) \dI x \,.
\end{equation}
We apply this for $\gamma = Q^{++}$ and $\gamma = -Q^{--}$ and  obtain
\begin{equation}
\tr (- \Lap_\lLT - \mu) Q^{\pm\pm} \gtrsim \int_{\BLT} S \lk ((|\rho^{\pm\pm}(x) |-C \lLT^{-1} \mu)_+ \rk) \dI x \,.  \label{eq:Part++--} 
\end{equation}

\begin{proofpart}{$Q^{+-}$, $Q^{-+}$}
\end{proofpart}
In the next step we want to prove bounds for $Q^{+-}$ and $Q^{-+}$. We introduce
\begin{align}\nonumber
\Pi_0^+ &= \unit(\mu < - \Lap_\lLT < \mu + \sqrt{\mu}/ \lLT) & \Pi_0^- & = \unit(\mu - \sqrt{\mu}/ \lLT \le - \Lap_\lLT \le \mu) \\
\Pi_1^+ &= \unit(\mu + \sqrt{\mu}/ \lLT \le - \Lap_\lLT) & \Pi_1^- & = \unit(- \Lap_\lLT < \mu  - \sqrt{\mu}/ \lLT)  
\end{align}
and split $Q^{+-} = (\Pi_0^+ + \Pi_1^+) Q (\Pi_0^- + \Pi_1^-) = Q^{+-}_{00} + Q^{+-}_{10} + Q^{+-}_{01} + Q^{+-}_{11}$. 
The following three parts of the proof will treat these terms. We start with $Q_{00}^\pm$. 
\begin{proofpart}{$Q^{+-}_{00}$}
\end{proofpart}
The density of $Q^{+-}_{00}$ is equal to
\begin{equation}
\rho^{+-}_{00}(x) = 
\sum_{\substack{k \in (\pi \mathbb N/\lLT)^3 \\ \mu  < k^2  < \mu + \sqrt \mu/\lLT}} \ \sum_{\substack{j \in (\pi \mathbb N/\lLT)^3 \\ \mu - \sqrt{\mu}/\lLT \le j^2 \le \mu}}
  \braket{\phi_k}{Q \phi_j} \phi_k(x) \phi_j(x) \,.
\end{equation}
Using $\norm Q \le 1$, we can bound this as 
\begin{align}\nonumber
|\rho^{+-}_{00}(x)| 
&\le \Big ( \sum_{ \mathclap{\substack{k \in (\pi \mathbb N/\lLT)^3 \\ \mu  < k^2  < \mu + \sqrt \mu/\lLT}}}|\phi_k(x)|^2 \Big )^{1/2} \Big ( \sum_{\mathclap{\substack{j \in (\pi \mathbb N/\lLT)^3 \\ \mu - \sqrt{\mu}/\lLT \le j^2 \le \mu}}} |\phi_j(x)|^2 \Big )^{1/2} \\ 
&\le \lk(\frac {2}{\lLT} \rk)^3 \sqrt{|\{\mu \le k^2 \le \mu + \sqrt{\mu}/\lLT \} |} \sqrt{|\{\mu - \sqrt{\mu}/\lLT \le j^2 \le \mu  \} |}   \le C \frac {\mu}{\lLT} \label{eq:Part+-01}
\end{align}
where we applied  \eqref{eq:diffCounting} in the last step.

\begin{proofpart}{$Q^{+-}_{10}, Q^{+-}_{01}$}
\end{proofpart}
Next we will bound $\rho_{10}^{+-}$.
 For a general function $W$ (viewed as a multiplication operator), 
 we have
\begin{align}\nonumber
|\tr(W Q_{10}^{+-})| & = \lk |\tr \lk ( \Pi_0^- W \frac {\Pi_1^+}{|- \Lap_\lLT - \mu|^{1/2}} |- \Lap_\lLT - \mu|^{1/2} Q \rk) \rk | \\  
& \le \sqrt{\tr|- \Lap_\lLT - \mu| Q^2}\norm {\Pi_0^- W \frac {\Pi_1^+}{|-\Lap_\lLT - \mu|^{1/2}}}_{\mathfrak S^2} \,. \label{eq:FT1} 
\end{align}
To bound the first factor, we can used Lemma~\ref{lem:Q2Q}.
For the second term we need to use the specific form of the eigenfunctions for the Dirichlet Laplacian. 
Using \eqref{eq:DirichletState} we get
\begin{equation}
|\braket{\phi_p}{W \phi_q}|^2 =  \lk ( \frac {1} {2 \lLT} \rk )^{6} \lk |\sum_{A,B \in \{1,-1\}^3} \hat W( (A_j p_j)_j- (B_j q_j)_j) \rk |^2 \lesssim \lLT^{-6}  \sum_{A,B \in \{1,-1\}^3} |\hat W( (A_j p_j)_j- (B_j q_j)_j)|^2 \label{eq:boundEigenvalueDirichlet} 
\end{equation}
where $(A_j p_j)_j$ and $(B_j q_j)_j$ denote the vectors obtained by component-wise multiplication.
Hence 
\begin{align}\nonumber 
\norm{ \Pi_0^- W \frac {\Pi_1^+}{|- \Lap_\lLT - \mu|^{1/2}} }_{\mathfrak S^2}^2 &= 
\sum_{\substack{p,q \in (\pi \mathbb N/\lLT)^3\\ \mu - \sqrt{\mu}/\lLT \le p^2 \le \mu \\ q^2 > \mu + \sqrt{\mu}/\lLT}} \frac {|\braket{\phi_p}{W \phi_q}|^2}{q^2 - \mu} \le 
\frac {\lLT}{\sqrt \mu} \sum_{\substack{p,q \in (\pi \mathbb N/\lLT)^3\\  \mu - \sqrt{\mu}/\lLT \le p^2 \le \mu \\ q^2 > \mu + \sqrt{\mu}/\lLT}} |\braket{\phi_p}{W \phi_q}|^2 \\ \nonumber
&\lesssim \frac 1 {\lLT^{6}} \frac {\lLT}{\sqrt \mu} \sum_{\substack{p,q \in (\pi (\mathbb Z \setminus \{0\})/\lLT)^3\\\mu - \sqrt{\mu}/\lLT \le p^2 \le \mu \\ q^2 > \mu + \sqrt{\mu}/\lLT}} |\hat W(p-q)|^2 \\
&\lesssim \frac 1 {\lLT^{6}} \frac \lLT {\sqrt{\mu}} \sum_{q\in (\pi (\mathbb Z \setminus \{0\})/\lLT)^3} |\hat W(q)|^2 \sum_{\mu - \sqrt{\mu}/\lLT \le p^2 \le \mu} 1 
\lesssim \sqrt \mu \norm{W}^2_2\,.
\end{align}
The sum of \eqref{eq:boundEigenvalueDirichlet} is included in the second line of the previous calculation by extending the sum over $p,q \in \mathbb N^3$ to $p,q \in (\mathbb Z \setminus \{0\})^3$, and we have again used  \eqref{eq:diffCounting} in the last step. 

Choosing for $W = (\rho_{10}^{+-})^*$ we thus get from \eqref{eq:FT1} 
\begin{equation}
\int_{\BLT} |\rho_{10}^{+-}|^2 \le C \mu^{1/2} \tr (- \Lap_\lLT  - \mu) Q  \,. \label{eq:Part+-1011} 
\end{equation}
In a similar way we can treat $\rho_{01}^{+-}$ with the result that also
\begin{equation}\label{a34}
\int_{\BLT} |\rho_{01}^{+-}|^2 \le C \mu^{1/2} \tr (- \Lap_\lLT  - \mu) Q  \,.
\end{equation}

\begin{proofpart}{$Q^{+-}_{11}$}
\end{proofpart}
Similarly to above we again introduce a multiplication operator $W$, and estimate
\begin{equation}
\left| \tr (W \Pi^{+}_1 Q \Pi^-_1) \right| 
\le \norm{ \frac {\Pi^+_1}{|\Lap_\lLT + \mu|^{1/4}} W \frac {\Pi^-_1}{|\Lap_\lLT + \mu|^{1/4}} }_{\mathfrak{S}^2} \norm{|\Lap_\lLT + \mu|^{1/4} Q |\Lap_\lLT + \mu|^{1/4}}_{\mathfrak{S}^2}  \,. 
\end{equation}
The second factor we bound by
\begin{equation}
\norm{|\Lap_\lLT + \mu|^{1/4} Q |\Lap_\lLT + \mu|^{1/4}}_{\mathfrak{S}^2} \le \norm{|\Lap_\lLT + \mu|^{1/2} Q }_{\mathfrak{S}^2} 
= \tr( |\Lap_\lLT + \mu| Q^2)^{1/2} \label{eq:Q10} 
\end{equation}
and Lemma \ref{lem:Q2Q}.
For the first one, we have
\begin{align}\nonumber
 & \norm{ \frac {\Pi^+_1}{|\Lap_\lLT + \mu|^{1/4}} W \frac {\Pi^-_1}{|\Lap_\lLT + \mu|^{1/4}} }_{\mathfrak{S}^2}^2 = \sum_{\substack{p,q \in (\pi \mathbb N^3/\lLT) \\p^2 > \mu + \sqrt{\mu}/\lLT\\q^2 < \mu - \sqrt{\mu}/\lLT}} \frac {|\braket{ \phi_p}{W \phi_q}|^2}{ (\mu - q^2)^{1/2} (p^2 - \mu)^{1/2}} \\ 
&  \le \frac C {\lLT^6} \sum_{\substack{p,q \in (\pi \mathbb Z^3/\lLT) \\p^2 > \mu + \sqrt{\mu}/\lLT\\q^2 < \mu - \sqrt{\mu}/\lLT}} \frac {|\hat W (q-p)|^2}{ (\mu - q^2)^{1/2} (p^2 - \mu)^{1/2}} = \frac {C}{\lLT^3} \sum_{k \in (\pi \mathbb Z^3/\lLT)} \Phi(k) |\hat W(k)|^2 \le C \sup_k \Phi(k) \norm{W}^2_2 \label{eq:+-01} 
\end{align}
with
\begin{equation}
\Phi(k) = \frac 1 {\lLT^3} \sum_{\substack{q \in (\pi \mathbb Z^3/\lLT) \\ (q-k)^2 > \mu + \sqrt{\mu}/\lLT \\ q^2 < \mu - \sqrt{\mu}/\lLT}} \frac 1 {(\mu - q^2)^{1/2}((q-k)^2-\mu)^{1/2}} \,.
\end{equation}
In \cite[Proof of Thm.~5.1]{posDensitiy} it was shown that 
$\sup_k \Phi(k) \lesssim \mu^{1/2}$ for $\mu\gtrsim L^{-2}$. 
Hence the choice $W=(\rho_{11}^{+-})^*$ yields 
\begin{equation}
\int_{\BLT} |\rho_{11}^{+-}|^2 \lesssim \mu^{1/2} \tr((-\Lap_\lLT -\mu) Q) \,. \label{eq:Part+-10} 
\end{equation}

\begin{proofpart}{Combining the above estimates}
\end{proofpart}
By combining \eqref{a34} and \eqref{eq:Part+-10} we obtain
\begin{equation}
\mu^{-1/2} \int_{\BLT} |\rho^{+-} - \rho_{00}^{+-}|^2 \le C \tr (- \Lap_\lLT - \mu) Q \,.
\end{equation}
Using that 
 $|\rho_{00}^{+-}| \le C \mu/\lLT$, as shown in \eqref{eq:Part+-01}, this further implies that 
\begin{equation}
\mu^{-1/2} \int_{\BLT} \lk(|\rho^{+-}| - C \mu/\lLT \rk)_+^2 \le C \tr (- \Lap_\lLT - \mu) Q \,.
\end{equation}
The integrand in the left side is bounded from below by $C S  ((|\rho^{+-}| - C \mu/\lLT)_+ )$, hence 
\begin{equation}
\int_{\BLT}  S \lk ((|\rho^{+-}| - C \mu/\lLT)_+ \rk ) \le C \tr (- \Lap_\lLT - \mu) Q \,. \label{eq:Part+-} 
\end{equation}
Since $|\rho^{+-}| = |\rho^{-+}|$, the same bound holds  for $\rho^{-+}$ as well.  
Combining \eqref{eq:Part++--} and \eqref{eq:Part+-} and using the convexity of $S$ we get
\begin{align}\nonumber
 \tr(-\Lap_\lLT - \mu) Q &\gtrsim  \int_{\BLT} S\lk (  \frac {\lk(|\rho^{++}|+ |\rho^{--}| + |\rho^{+-}| + |\rho^{-+}|-C \mu/\lLT\rk)_+} 4 \rk) \\
&\ge \int_{\BLT} S\lk (  \frac {\lk(|\rho_Q|-C \mu/\lLT\rk)_+} 4 \rk) \gtrsim  \int_{\BLT} S( (|\rho_{Q}|-C \mu/\lLT)_+)  \,.
\end{align}
This completes the proof of Theorem~\ref{thm:posDensity}.
\end{proof}

By taking a Legendre transform, the result above implies that following potential version of the Lieb-Thirring inequality.
\begin{theorem}
\label{thm:lti}
Assume that   $V$ is a real-valued function in $L^{5/2}([-\lLT/2,-\lLT/2]^3)$, and $\mu \ge 3 \pi^2 \lLT^{-2}$. Then we have
\begin{align}\nonumber 
0 &\ge - \tr (- \Lap_\lLT + V - \mu)_- + \tr(- \Lap_\lLT - \mu)_- -  \int_{C_L} \rho_0 V \\
&\ge - K \int_{C_L} \Big( \mu^{1/2} |V|^2 + |V|^{5/2} + \lLT^{-1} \mu |V|\Big)\label{eq:lti} 
\end{align}
with $K > 0$ independent of $\lLT, \mu$ and $V$. 
\end{theorem}

\begin{remark*}
In case that $\mu < 3 \pi^2 \lLT^{-2}$ we have $- \Lap_\lLT - \mu > 0$, and therefore  $\tr(- \Lap_\lLT - \mu)_- = 0$ and also  $\rho_0=0$. One can thus obtain a lower bound using the standard Lieb-Thirring inequality \cite{LT} applied to a potential $V - \mu$ in this case. 
\end{remark*}

\begin{proof}
We start with the identity 
\begin{equation}
- \tr(A + B)_- = \inf_{0\le \gamma \le 1} \tr(A + B) \gamma \label{eq:appProof1} 
\end{equation}
for hermitian matrices $A$ and $B$,  where an optimizer is clearly $\unit (A + B \le 0)$.  With $P^- = \unit(A \le 0)$ and $Q = \gamma  -P^-$, \eqref{eq:appProof1} reads
\begin{equation}
- \tr(A+B)_- = \inf_{-P^- \le Q \le 1 - P^-} \tr(A + B) Q + \tr(A + B) P^- \,.
\end{equation}
Defining $P_{B}^- = \unit(A+B \le 0)$ we equivalently get
\begin{equation}
\tr(A+B)(P^-_B - P^-) =  \inf_{-P^- \le Q \le 1 - P^-} \tr(A + B) Q \,.  \label{eq:inMatrix} 
\end{equation}
This equality can be extended to allow $A = - \Lap - \mu$ and $B = V$ (see \cite[Thm 4.1]{posDensitiy}). Using this and applying Theorem \ref{thm:posDensity} we get
\begin{align}\nonumber
\tr(-\Lap_\lLT - \mu)_- -\tr (-\Lap_\lLT + V - \mu)_- - \int_{\BLT} \rho_0 V &\ge \inf_{\rho} \lk ( \tilde K \int_{\BLT} S((|\rho| - \eta \lLT^{-1} \mu)_+) + \int_{\BLT} V \rho \rk ) \\
 &\ge \inf_{\rho\geq 0} \lk ( \tilde K \int_{\BLT} S((\rho - \eta \lLT^{-1} \mu)_+) - \int_{\BLT} |V| \rho \rk ) \label{eq:lb1} 
\end{align}
where the infimum in the first line is over functions $\rho: \R^3 \to \R$, while in the second we can restrict to non-negative functions $\rho$. We can pull the infimum inside the integral for a lower bound. Clearly we can assume that $\rho \geq \eta \lLT^{-1} \mu$. Introducing $\gamma = \rho - \eta \lLT^{-1} \mu$ we have
\begin{equation}
\inf_{\gamma \geq 0} \lk ( \tilde K S(\gamma) - |V| \gamma - \eta \lLT^{-1} \mu |V| \rk ) = \tilde K \lk ( \frac 2 3 \mu ^{5/2}+ \tilde K^{-1} |V| \mu^{3/2}- \frac 2 3 \left(\mu + \tilde K^{-1}  \frac{3 |V|}{5}\right)^{5/2}  \rk ) - \eta \lLT^{-1} \mu   |V| \,. \label{eq:gammapos}  
\end{equation}
Using that 
\begin{equation}
x^{5/2} + \frac 5 2 x^{3/2} y - (x+y)^{5/2}  \ge - \frac{15 \sqrt{x} y^2}{8} - y^{5/2} 
\end{equation}
for $x = \mu$ and $y = 3 \tilde K^{-1} |V|/5$
 gives the bound 
\begin{equation}
\eqref{eq:gammapos} \gtrsim - \mu^{1/2} |V|^2 - |V|^{5/2} - \lLT^{-1} \mu |V|   \,.
\end{equation}
Plugging this into \eqref{eq:lb1} proves the Theorem.
\end{proof}

We apply the above theorem for a potential $V \in L^{5/2}(\BLT)$ with $V \leq 0$, choosing $\mu$ as $e_N$, the $N$th eigenvalue of the Dirichlet Laplacian $- \Lap_\lLT$. In particular, $\mu \ge e_1 = 3 \pi^2 \lLT^{-2}$ which allows us to use Theorem \ref{thm:lti}. The ground state energy $\EGD$ for  $N$ non-interacting particles confined to $C_L$ was defined in the beginning of Section~\ref{sec:Results} and can be written as $\EGD = \sum_{i=1}^N e_i$.

We denote by $e_k^V$ the $k$th eigenvalue of $-\Lap_\lLT + V$, and by $\EGDV$ the sum of the lowest $N$ eigenvalues of $-\Lap_\lLT + V$, i.e., $\EGDV = \sum_{i=1}^N e_i^V$. Theorem \ref{thm:lti} implies that 
\begin{equation}
\tr(- \Lap_\lLT-\mu)_- = -\EGD + N e_N \ge \tr(-\Lap_\lLT + V - e_N)_- - R 
\ge  - \EGDV + N e_N - R   
\end{equation}
with
\begin{equation}
R = \const \int_{\BLT} \lk (  \mu^{1/2} |V|^2 + |V|^{5/2} + \lLT^{-1} \mu |V| \rk) - \int_{\BLT} \rho_0 V \,.
\end{equation}
We used that since $V \leq 0$ the operator $-\Lap_\lLT + V  -e_N$ has at least $N$ non-positive eigenvalues,  and therefore we can get a lower bound on the trace of its negative part by summing only the first $N$ of them.

From the above calculation, together with $\rho_0 \lesssim \mu^{3/2}$ and $\mu= e_N \lesssim N^{2/3}/\lLT^2$, we deduce the following corollary.
\begin{corollary}
\label{cor:BoundPotential}
Let $V \in L^{5/2}(\BLT)$ with $V\leq 0$ and let $\EGD$ denote the ground state energy of $N$ non-interacting fermions  confined to $C_L$. With $\EGDV$ we denote the ground state energy of the corresponding Hamiltonian with  external potential $V$. Then 
\begin{equation}\label{eq:cora1}
\EGD - \EGDV \lesssim \int_{\BLT} \lk(  \frac {N^{1/3}} \lLT |V|^2 + |V|^{5/2} + \frac {N}{\lLT^3} |V| \rk)  \,.
\end{equation}
\end{corollary}

\section*{Acknowledgments}
We would like to thank Ulrich Linden for many helpful  discussions. 
Financial support by the European Research Council (ERC) under the European
Union's Horizon 2020 research and innovation programme (grant agreement No
694227), and by the Austrian Science Fund (FWF), project Nr. P~27533-N27, is
gratefully acknowledged.

\end{document}